\documentclass[a4paper,12pt]{article}

\usepackage[T2A]{fontenc}

\usepackage[english]{babel}

\usepackage{graphicx}
\usepackage{amsthm, amsmath, amssymb, amsbsy, amsfonts, dsfont}
\usepackage{srcltx, color}

\def\tht{\theta}
\def\Om{\Omega}
\def\om{\omega}
\def\e{\varepsilon}
\def\g{\gamma}

\def\l{\lambda}
\def\p{\partial}
\def\D{\Delta}

\def\a{\alpha}
\def\b{\beta}

\def\d{\delta}

\def\z{\zeta}

\def\vp{\varphi}

\def\vr{\varrho}

\def\H{W_2}

\def\iu{\mathrm{i}}
\def\di{\,\mathrm{d}}

\def\Op{\mathcal{H}}

\def\x{\mathrm{x}}

\def\hf{\mathfrak{h}}

\def\la{\langle}
\def\ra{\rangle}

\def\H{W_2}

\def\cB{\mathcal{B}}

\def\cP{\mathcal{P}}

\def\Dom{\mathfrak{D}}

\DeclareMathOperator{\spec}{\sigma}

\DeclareMathOperator{\RE}{Re}
\DeclareMathOperator{\IM}{Im}
\DeclareMathOperator{\dist}{dist}

\DeclareMathOperator{\tr}{tr}

\numberwithin{equation}{section}

\newtheorem{theorem}{Theorem}[section]
\newtheorem{lemma}{Lemma}[section]

\usepackage[normalem]{ulem}

\begin{document}

\allowdisplaybreaks

\title{Approximation of point interactions by geometric perturbations in two-dimensional domains}

\date{\empty}

\author
{D.I. Borisov$^1$\footnote{Corresponding author}, P. Exner$^2$}

\vskip -0.5 true cm

\maketitle

\begin{center}
{\footnotesize $^1$
Institute of Mathematics, Ufa Federal Research Center, Russian Academy of Sciences, Ufa, Russia,
\\
Bashkir State Pedagogical
University named after M.~Akhmulla,
Ufa, Russia,
\\
University of Hradec Kr\'alov\'e,  Hradec Kr\'alov\'e, Czech Republic
 \\
{\tt borisovdi@yandex.ru}
\\
\footnotesize $^2$
Doppler Institute for Mathematical Physics and Applied Mathematics, Czech
Technical University in Prague, B\v{r}ehov\'{a} 7, 11519 Prague
\\
Nuclear Physics Institute, Czech Academy of Sciences, 25068 \v{R}e\v{z} near Prague, Czech Republic
\\[-.3em]
{\tt exner@ujf.cas.cz}
}
\end{center}

\begin{abstract}
We present a new type of approximation  of a second-order elliptic operator in a planar domain with a point interaction. It is of a geometric nature, the approximating family consists of operators with the same symbol and regular coefficients on the domain with a small hole. At the boundary of it Robin condition is imposed with the coefficient which depends on the linear size of a hole. We show that as the hole shrinks to a point and the parameter in the boundary condition is scaled in a suitable  way, nonlinear and singular, the indicated family converges in the norm-resolvent sense to the operator with the point interaction. This resolvent convergence is established with respect to several operator norms and order-sharp estimates of the convergence rates are provided.
\end{abstract}

\section{Introduction}

Operators with singular, point-like perturbations attracted attention in the early days of quantum mechanics as idealized models for Hamiltonians of systems in which the interaction is concentrated in a small area \cite{Fe36}. The advantage of such an idealized description is that one can simplify considerably spectral analysis of such operators. From the mathematical point of view, point interactions are easy to deal with in the case of ordinary differential operators where they are described by appropriate boundary conditions. In the practically important cases of dimensions two and three the question is more difficult, however, and it took time before Berezin and Faddeev \cite{BF61} showed how to describe these operators in terms of self-adjoint extensions. In the recent decades point interactions were a subject of an intense interest; we refer to the monograph \cite{AGHH} for the presentation of the theory and an extensive bibliography.

The key thing in application of the point interaction models is to understand how they can be approximated by operators with regular coefficients. This is again easy in case of the ordinary differential operators\footnote{This claim applies to the so-called $\delta$ potentials, there are more singular point interactions in one dimension for which the approximation is a far more complicated matter, see e.g. \cite{AN00, ENZ01}.} describing systems in one spatial dimension, where such an interaction is the limit of naturally scaled potentials. In dimensions two and three the procedure is much more subtle because scaling of the coefficients leads generically to a trivial result. One has to use a particular way of nonlinear scaling starting from the situation when the initial operator has a spectral singularity at the threshold of the continuous spectrum; a physicist would speak about a particular way of `coupling constant renormalization'.

With the importance of the point interaction models in mind, it would be certainly useful to have approximations other than the standard one mentioned above and described in \cite[Sec.I.1 and I.5]{AGHH}. The aim of the present paper is to present an alternative approximation to two-dimensional point interactions, which is of a geometric nature. It employs families of operators with the same differential expression as the unperturbed one but restricted to the exterior of a small hole containing the support of the point interaction; at the boundary of the hole we impose Robin boundary condition with a coefficient depending in a singular way on a parameter characterizing the linear size of the hole. Shrinking the hole and scaling properly the parameter in the boundary condition, we obtain an operator family that converges, in the norm-resolvent sense, to an operator with a point interaction in the domain without the hole. The convergence is established in terms of several operators norms and for  each of them, we obtain  order-sharp  estimates for the convergence rate. As a consequence, we also obtain the convergence of the operator spectra and the associated spectral projectors.

It should be noted that elliptic boundary value problems with small holes represent a classical example in the singular perturbation theory. Situations when the boundary of the hole is subject to one of the classical boundary conditions were investigated, for instance, in \cite{Il, MNP1984}, where typically a weak or strong resolvent convergence was established. Asymptotic expansions for solutions to such problems, in the first place, for the corresponding eigenvalues and eigenfunctions, were found under appropriate smoothness assumptions. Recent results on norm-resolvent convergence in the boundary homogenization theory \cite{PRSE, ZAMP, JDE} inspired results on the same convergence for operators in domains with small holes \cite{BM2018, IMM12}, however, in these  papers  a fixed classical boundary conditions was always imposed at the boundary of the hole, in particular, the Robin condition was used with the coefficient independent of the hole size.

As we have mentioned, we work in the two-dimensional setting. The way we present our result is particular and general at the same time. The particularity reflects the fact that we deal with approximation of a single point interaction, and moreover, that our result also has a limitation: our approximation applies only to point interactions  which are, roughly speaking, attractive enough in the sense made precise by the condition \eqref{2.4a} below;  recall that, for instance, a single point-interaction perturbation of the Laplacian in the plane is always attractive \cite[Sec.~I.5]{AGHH}. On the other hand, our proof is of a local nature and there would no problem to extend it to cover a finite number of point interactions; each of them will be approximated by an appropriate hole with Robin boundary and all estimates in Theorem~\ref{th:main} would remain true, even if the involved expression would be pretty cumbersome. It is also possible to consider operators with infinitely many point interactions provided the mutual distances between their supports have a positive lower bound, but then additional restrictions on the coefficients in the differential expression would be needed.

What is more important, in contrast to standard treatment as one can find in the monograph
 \cite{AGHH}, our analysis is more general in two respects. First of all, we consider point perturbations in an arbitrary planar domain and, if such a domain is not the entire plane, any (local) boundary conditions can be chosen for the free operator. In particular, since  the boundary operator $\cB$ can also describe quasi-periodic boundary conditions, our result covers, in the usual Floquet way, infinite periodic systems of point interactions with a single perturbation in the period cell without any additional assumptions.

Secondly, our unperturbed operator is not just a Laplacian or a Schr\"odinger operator, but a general second-order elliptic operator; in Sec~\ref{ss:SA0} we define a point perturbation of such an operator properly and show that it is self-adjoint. Our results thus  allows us to treat singular perturbations of more general systems such as magnetic Schr\"odinger operator or Hamiltonians with a weight in the kinetic term, in other words, systems with a position-dependent `mass'. This could be of interest in solid state physics, where the effective electron mass depends on the material and becomes nontrivial in composite structures build, say, from different types of semiconductors. In such models, the hole in the perturbed problem can be interpreted as a localized defect in the material with  a particular surface interaction at its boundary.

\section{Statement of the problem and the results}

Let $x=(x_1,x_2)$ be Cartesian coordinates in $\mathbb{R}^2$ and $\Om\subseteq\mathbb{R}^2$ be a domain which can be both bounded or unbounded, including the particular case of $\Om=\mathbb{R}^2$. If the boundary of $\Om$ is nonempty, we assume that it is $C^2$-smooth.

By $x_0$ we denote an arbitrary fixed point of $\Om$ and consider its neighborhood of which we will speak as of a \emph{hole}, defined as $\om_\e:=\{x:\, (x-x_0)\e^{-1}\in\om\}$, where $\om\subset\mathbb{R}^2$ is a bounded simply connected set the boundary of which is $C^3$-smooth. The hole is supposed to be small, its size being controlled by the positive parameter $\e$, and we assume that $\om$ contains the origin of the coordinates so that $x_0\in\om_\e$ for all $\e>0$.

The main object of our interest is the family of self-adjoint scalar second-order differential operators
\begin{equation}\label{2.1}
\Op_\e=-\sum\limits_{i,j=1}^{2} \frac{\p\ }{\p x_i} A_{ij} \frac{\p\ }{\p x_j} +  \iu\sum\limits_{j=1}^{2}\left( A_j \frac{\p\ }{\p x_j} + \frac{\p\ }{\p x_j} A_j\right) + A_0 \quad\;\text{in}\quad\; \Om_\e:= \Om\setminus\om_\e
\end{equation}
subject to one of the classical, $\e$-independent boundary conditions on $\p\Om$,
\begin{equation}
\label{2.2b}
\cB u=0\quad\text{on}\quad\p\Om,
\end{equation}
and to the Robin condition on $\p\om_\e$ that scales singularly with respect to $\e$ as follows,
\begin{align}
&\frac{\p u}{\p \mathrm{n}}=\frac{\a\big(\e^{-1}s_\e,\ln^{-1}\e\big)}{\e\ln\e}u \quad\text{on}\quad\p\om_\e,\label{2.2a}
\\[.5em]
& \a(s,\mu):=\a_0(s)+\mu \a_1(s).\label{2.2c}
\end{align}
The operator $\cB$ in \eqref{2.2b} can be arbitrary. For instance, $\cB u=u$ refers to Dirichlet condition and $\cB u=\frac{\p u}{\p\mathrm{n}}+b_0 u$ describes Robin condition with the parameter $b_0$. Another option is represented by quasi-periodic boundary conditions, and any combination of these conditions on different subsets of $\p\Om$ is also admissible.

The coefficients $A_{ij}=A_{ij}(x)$, $A_j=A_j(x)$, and $A_0=A_0(x)$ in \eqref{2.1} are real functions on the closure $\overline{\Om}$. We assume that $A_{ij}, A_j\in C^3(\overline{\Om})$, $A_0\in C^2(\overline{\Om})$, and the functions $A_{ij}$ satisfy the standard ellipticity condition
\begin{equation}\label{2.12}
A_{ij}=A_{ji},\quad\; \sum\limits_{i,j=1}^{2}A_{ij}(x)\xi_i\xi_j\geqslant c_0 (\xi_1^2+\xi_2^2),\quad\; \xi_i\in\mathbb{R}, \;\; x\in\overline{\Om},
\end{equation}
where $c_0$ is a fixed positive constant independent of $x$ and $\xi$. Furthermore, by $\frac{\p\ }{\p\mathrm{n}}$ we denote the conormal derivative,
\begin{equation}\label{2.13}
\frac{\p\ }{\p\mathrm{n}}:=\sum\limits_{i,j=1}^{2}A_{ij} \nu_i \frac{\p\ }{\p x_i}-\iu\sum\limits_{j=1}^{2} \nu_j A_j,
\end{equation}
where $\nu=(\nu_1,\nu_2)$ is the unit normal on $\p\om_\e$ pointing inside $\om_\e$, and $\,\iu\,$ is the imaginary unit. The symbols $\a_0=\a_0(s)$, $\a_1=\a_1(s)$ stand for real functions on $\p\om$ continuous with respect to the arc length $s\in[0,|\p\om|]$. Similarly $s_\e$ denotes the arc length of $\p\om_\e$ for which $s_\e=\e s$ naturally holds. If $\p\Om$ is empty, then condition (\ref{2.2b}) is simply omitted, and the same applies hereafter to all the conditions imposed on $\p\Om$.

The aim of the present paper is to investigate
the resolvent convergence of the operators $\Op_\e$ as the scaling parameter $\e$ tends to zero.

Before stating our main result, we need to introduce some more notations. By $\Op_\Om$ we denote the operator in $L_2(\Om)$ with the differential expression $\hat{\Op}$ given by the right hand side in (\ref{2.1}) and subject to boundary condition (\ref{2.2b}). Furthermore, it  follows from the definition of the hole $\om_\e$ that there exist positive constants $R_1$, $R_2$ independent of $\e$ such that
\begin{equation}\label{3.6}
\om_\e\subset B_{R_1\e}(x_0) \subset B_{2R_1\e}(x_0) \subset B_{R_2}(x_0) \subset B_{2R_2}(x_0) \subset \Om_0 \subset \Om,
\end{equation}
where $B_r(a)$ denotes conventionally the disc of radius $r$ centered at the point $a$ and $\Om_0$ is the set specified in the following paragraph.

We adopt the following assumptions on the coefficients $A_{ij}$, $A_j$, $A_0$ in \eqref{2.1}, on those specifying the operator $\cB$ in \eqref{2.2b}, and on the operator $\Op_\Om$. The latter is supposed to be self-adjoint in $L_2(\Om)$ and semibounded from below, the associated closed symmetric sesquilinear form being denoted by  $\hf_\Om$. The domain $\Dom(\hf_\Om)$ is a subspace in $\H^1(\Om)$, and moreover, there exists a domain $\Om_0\subset\Om$ containing $x_0$ such that the restriction of each function from the domain $\Dom(\Op_\Om)$ on $\Om_0$ belongs to $\H^2(\Om_0)$. The form $\hf_\Om$ satisfies the following lower bound
\begin{equation}\label{2.0a}
\hf_\Om[u]-\hf_{\Om_0}[u]+c_1\|u\|_{L_2(\Om\setminus\Om_0)}^2\geqslant c_2 \|u\|_{\H^1(\Om\setminus\Om_0)}^2
\end{equation}
for all $u\in\Dom(\hf_\Om)$, where $c_1,\,c_2$ are positive constants independent of $u$. More generally, for an arbitrary subdomain $\tilde{\Om}\subset\Om$ and vectors $u,v\in\Dom(\tilde{\Om})$  we denote
\begin{equation}\label{2.0b}
\begin{aligned}
\hf_{\tilde{\Om}}(u,v):=&
\sum\limits_{i,j=1}^{2} \left(A_{ij}\frac{\p u}{\p x_j}, \frac{\p v}{\p x_i}\right)_{L_2(\tilde{\Om})}
+ \iu \sum\limits_{j=1}^{2}  \left(\frac{\p u}{\p x_j},A_j v\right)_{L_2(\tilde{\Om})}
\\
&- \iu \sum\limits_{j=1}^{2}  \left(A_j u,\frac{\p v}{\p x_j}\right)_{L_2(\tilde{\Om})} + (A_0u,v)_{L_2(\tilde{\Om})}.
\end{aligned}
\end{equation}
If $\tilde{\Om}$ has a positive distance from $\p\Om$, this form satisfies the lower bound
\begin{equation}\label{2.0c}
\hf_{\tilde{\Om}}[u]+c_1\|u\|_{L_2(\tilde{\Om})}^2\geqslant c_2 \|u\|_{\H^1(\tilde{\Om})}^2
\end{equation}
with the same constants $c_1$, $c_2$ as in (\ref{2.0a}).

To define the operator $\Op_\e$ rigorously, we use an infinitely differentiable cut-off function $\chi_\Om$ taking values in $[0,1]$, equal to one in $B_{2R_2}(x_0)$, and vanishing outside $\Om_0$. Then $\Op_\e$ is the operator in   $L_2(\Om_\e)$ with the differential expression $\hat{\Op}$ and the domain $\Dom(\Op_\e)$ consisting of the functions $u$ 
satisfying condition (\ref{2.2a}) and such that
\begin{equation}\label{2.16}
(1-\chi_\Om)u\in \Dom(\Op_\Om), \quad\;
 \chi_\Om u\in\H^2(\Om_0\setminus\om_\e)\,;
\end{equation}
the action of $\Op_\e$ is then given by the formula
\begin{equation}\label{2.16a}
\Op_\e u:=\Op_\Om(1-\chi_\Om)u + \hat{\Op} \chi_\Om u.
\end{equation}

Next we have to specify the limit of the operator family $\{\Op_\e\}_{\e>0}$. Referring to Section~\ref{ss:aux} below, in Lemma~\ref{lm:G} we will establish the existence of a unique solution $G\in W_2^2(\Om\setminus B_\d(x_0))\cap C^1(\overline{B_\d}\setminus\{x_0\})$, $\d>0$, to the boundary-value problem
\begin{equation}\label{2.6}
(\hat{\Op}+c_1)G=0\quad\text{in}\quad\Om\setminus\{x_0\},\quad\; \cB G=0\;\;\text{on}\;\; \p\Om,
\end{equation}
where $c_1$ is the constant from \eqref{2.0a} and \eqref{2.0c}, that behaves in the vicinity of $x_0$ as follows,
\begin{equation}\label{2.7}
G(x)=\ln|\mathrm{A}^{-\frac{1}{2}} (x-x_0)|+a+O\big(|x-x_0|\ln|x-x_0|\big),\qquad x\to x_0,
\end{equation}
with $a\in\mathbb{R}$ being a fixed number, $\mathrm{E}$ is the unit $2\times 2$ matrix and
\begin{equation*}
\mathrm{A}:=
\begin{pmatrix}
A_{11}(x_0) & A_{12}(x_0)
\\
A_{21}(x_0) & A_{22}(x_0)
\end{pmatrix}.
\end{equation*}
By $\x=\x(s)$ we denote the vector equation of the boundary, that is, the curve $\x:[0,|\p\om|]\to\Om$ coincides
with $\p\om$. We put
\begin{equation}\label{2.4}
  \a_0(s)=\frac{\nu\cdot\mathrm{A}^{\frac{1}{2}}\x(s)} {|\mathrm{A}^{-\frac{1}{2}}\x(s)|^2},
\end{equation}
suppose that $\a_1$ is such that
\begin{equation}
  K:=-\int\limits_{\p\om} \big(\a_0(s) \ln|\mathrm{A}^{-\frac{1}{2}}\x(s)|
+\a_1(s)\big)\mathrm{d}s>-c_2\|G\|_{L_2(\Om)}^2-\pi a\tr\mathrm{A}
\label{2.4a}
\end{equation}
holds, and denote
\begin{equation}\label{2.5}
\b:=-\frac{K}{\pi\tr\mathrm{A}},
\end{equation}
assuming in addition that $\b\ne a$.

The limiting operator of the family $\{\Op_\e\}_{\e>0}$ turns out to be the operator with the differential expression $\hat{\Op}$ in $\Om$ and a point interaction at the point $x_0$. We denote it $\Op_{0,\b}$; it is an operator in $L_2(\Om)$ with the domain
\begin{equation}\label{2.15}
\Dom(\Op_{0,\b}):=\left\{u=u(x):\, u(x)=v(x)+(\b-a)^{-1} v(x_0) G(x),\ v\in \Dom(\Op_\Om)\right\}
\end{equation}
acting as
\begin{equation}\label{2.10}
\Op_{0,\b} u=\Op_\Om v - c_1(\b-a)^{-1} v(x_0) G,
\end{equation}
where $c_1$ is again the constant from \eqref{2.0a} and \eqref{2.0c}

By $\|\cdot\|_{X\to Y}$ we denote the norm of a bounded operator acting from a Hilbert space $X$ into a Hilbert space $Y$. Now we are in position to state our main result:

\begin{theorem}\label{th:main}
The operators $\Op_\e$ and $\Op_{0,\b}$ are self-adjoint and $\Op_\e$ converges to $\Op_{0,\b}$ in the norm resolvent sense as $\e\to+0$. Namely, the following estimates hold,
\begin{align}\label{cnv1}
&\|(\Op_\e-\l)^{-1}-(\Op_{0,\b}-\l)^{-1}\|_{L_2(\Om)\to  L_2(\Om_\e)} \leqslant C|\ln\e|^{-1},
\\[.3em]
&\big\|\nabla\big((\Op_\e-\l)^{-1}-(\Op_{0,\b}-\l)^{-1}\big) \big\|_{L_2(\Om)\to  L_2(\Om_\e)}\leqslant C|\ln\e|^{-\frac{1}{2}},\label{cnv2}
\\[.3em]
&\big\|\chi_{\tilde{\Om}}\big((\Op_\e-\l)^{-1}-(\Op_{0,\b}-\l)^{-1}\big)\big\|_{L_2(\Om)\to  \Dom(\hf_\Om)}\leqslant C|\ln\e|^{-1},\label{cnv3}
\end{align}
where $\tilde{\Om}$ is an arbitrary fixed subdomain of $\Om$ such that $x_0\notin\tilde{\Om}$ and $\chi_{\tilde{\Om}}$ is an infinitely differentiable cut-off function equal to one on $\tilde{\Om}$ and vanishing outside some fixed domain containing $\tilde{\Om}$, still separated from the point $x_0$ by a positive distance. These estimates are order-sharp; the positive constants $C$ are independent of $\e$, the constant in estimate (\ref{cnv3})  may in general depend on the choice of $\tilde{\Om}$.
\end{theorem}
Our second main results describes the spectral convergence  of  the  operators $\Op_\e$; the spectrum of an operator is denoted by $\spec(\cdot)$.

\begin{theorem}\label{th2.2}
The spectrum of the operator $\Op_\e$ converges to  that  of $\Op_{0,\b}$  as $\e\to+0$. More specifically, if $\l\notin\spec(\Op_{0,\b})$, then $\l\notin\spec(\Op_\e)$ provided $\e$ is small enough, while if $\l\in\spec(\Op_{0,\b})$, then there exists a point $\l_\e\in\spec(\Op_\e)$ such that $\l_\e\to\l$ as $\e\to+0$. For any  $\vr_1, \vr_2\notin\spec(\Op_{0,\b})$, $\vr_1<\vr_2$, the spectral projection of $\Op_\e$ corresponding to the segment $[\vr_1,\vr_2]$ converges to the spectral projection of $\Op_{0,\b}$ referring  to the same segment in the sense of the norm $\|\cdot\|_{L_2(\Om)\to L_2(\Om_\e)}$.

For each fixed segment $Q:=[\vr_1,\vr_2]$  of  the real line the inclusion
\begin{equation}\label{2.17}
\spec(\Op_\e)\cap Q\subset
\big\{\l\in Q:\,\dist(\l,\spec(\Op_{0,\b})\cap Q)\leqslant C|\ln\e|^{-1} \}
\end{equation}
holds, where $C$ is a fixed constant independent of $\e$ but depending of $Q$. If $\l_0$ is an isolated  eigenvalue of $\Op_{0,\b}$ of a multiplicity $n$, there exist exactly $n$ eigenvalues of the operator $\Op_\e$, counting multiplicities, which converge to $\l_0$ as $\e\to+0$. The total projection $\cP_\e$ referring to these perturbed eigenvalues and the projection  $\cP_{0,\b}$ onto the eigenspace associated with $\l_0$ satisfy   estimates analogous to (\ref{cnv1}), (\ref{cnv2}), and (\ref{cnv3}).
\end{theorem}

Before proceeding to  the theorems, let us add a few comments. The convergence of  $\Op_\e$ to $\Op_{0,\b}$  is expressed in terms of several norms for the corresponding difference of the resolvents,  namely those of operators acting from $L_2(\Om)$ into $L_2(\Om_\e)$ or $\H^1(\Om_\e)$, see (\ref{cnv1}), (\ref{cnv2}). One more estimate is given in (\ref{cnv3}), where the norm involves a cut-off function $\chi_{\tilde{\Om}}$. The presence of this cut-off function means that the difference of the resolvents is considered on a fixed subdomain separated from the point $x_0$; this difference is estimated in the norm defined by the form of the operator $\Op_\Om$. The convergence rates in (\ref{cnv1}), (\ref{cnv3}) are same being $O(|\ln\e|^{-1})$, while the rate in (\ref{cnv2}) is just $O(|\ln\e|^{-\frac{1}{2}})$. The reason is that the norm in (\ref{cnv2}) is stronger  than in (\ref{cnv1}) since it involves the gradient; note that in (\ref{cnv3}) its presence plays no role, because the norm is considered on the domain separated from the point $x_0$.

As indicated in the introduction, the  constant $\b$ defined by (\ref{2.5}) can not take all values on the real line  in view of (\ref{2.4a}).  This condition obviously fixes an upper bound for the admissible values of $\b$ and, at the same time, it is  essential for our technique; should (\ref{2.5}) fail, the convergence of our operator families could fail as well.

Our second result, Theorem~\ref{th2.2}, states the convergence of the spectrum and the associated spectral projections. This result is based  essentially  on standard theorems about  the convergence of the spectra  with respect to the resolvent norm, however, they can not be applied directly here since the operators $\Op_\e$ and $\Op_{0,b}$ act on different spaces. One more problem is that the functions in the domain of the limiting operator exhibit a logarithmic singularity at $x_0$. Nevertheless, we succeed to  overcome these obstacles. Moreover, inclusion (\ref{2.17}) provides, in fact, an estimate for the convergence rate of the spectrum, which turns out to be the same as in inequality (\ref{cnv1}). Indeed, this inclusion means that once we consider compact parts of the spectra of $\Op_{0,\b}$ and $\Op_\e$, the distance between the perturbed spectrum and the limiting one is of order $O(|\ln\e|^{-1})$. Considering then how the isolated eigenvalues of the operator $\Op_{0,\b}$ bifurcate into the eigenvalues of $\Op_\e$, we are able also to estimate the convergence rate for the associated spectral projections arriving at  estimates  that  are the same as (\ref{cnv1}), (\ref{cnv2}), (\ref{cnv3}).


\section{Auxiliary results}\label{ss:aux}

Here we collect several auxiliary results, which will help us to prove Theorem~\ref{th:main} in the next section.

\begin{lemma}\label{lm3.0}
The identity
\begin{equation}
\label{5.5}
\int\limits_{\p\om} \a_0(s)\,\mathrm{d}s =
\int\limits_{\p\om} \frac{\nu\cdot\mathrm{A}^\frac{1}{2}\x(s)} {|\mathrm{A}^{-\frac{1}{2}}\x(s)|^2}\,\mathrm{d}s=-\pi\tr\mathrm{A}
\end{equation}
holds true.
\end{lemma}
\begin{proof}
Let us express the integral on the left-hand side of (\ref{5.5}).
We observe that
\begin{equation*}
\sum\limits_{i,j=1}^{2} A_{ij}(x_0)\frac{\p^2\ }{\p x_i\p x_j} \ln |\mathrm{A}^{-\frac{1}{2}}(x-x_0)|=0
\end{equation*}
holds in the vicinity of $x_0$. To see that this the case, one can introduce local coordinates, $y:=\mathrm{A}^{-\frac{1}{2}}(x-x_0)$, in which the expression in question is nothing else than $\Delta\ln|y|$. Integrating it over $\omega$ with a small disc centered at $x_0$ deleted, using Green's formula, we get
\begin{equation*}
0=\!\!\int\limits_{\om\setminus\{x:\, |y|<\d\}}  \sum\limits_{i,j=1}^{2} A_{ij}(x_0)\frac{\p^2\ }{\p x_i\p x_j} \ln |\mathrm{A}^{-\frac{1}{2}}(x-x_0)|\,\mathrm{d}x
=-\int\limits_{\p\om} \frac{\p\ }{\p\mathrm{n}}\ln |\mathrm{A}^{-\frac{1}{2}}\x(s)|\,\mathrm{d}s -\!\!\int\limits_{\{x:\,|y|=\d\}} \!\!\frac{\mathrm{d}s}{|y|}
\end{equation*}
Evaluating the integrals on the right-hand side and taking the limit $\delta\to+0$ in the second one, we find
\begin{equation*}
0=-\int\limits_{\p\om} \frac{\nu\cdot\mathrm{A}^\frac{1}{2}\x(s)\,ds} {|\mathrm{A}^{-\frac{1}{2}}\x(s)|^2}\,\mathrm{d}s-\pi\tr\mathrm{A},
\end{equation*}
in other words, the sought identity (\ref{5.5}).
\end{proof}

In view of the assumptions made about the operator $\Op_\Om$, in particular, of the estimates (\ref{2.0a}) and (\ref{2.0c}), the spectrum of this operator is contained in the interval $[c_2-c_1,\infty)$, and since $c_2>0$, the inverse operator $(\Op_\Om+c_1)^{-1}$ is well-defined and bounded. In the following lemma we employ the polar coordinates $(r,\tht)$ associated with the variables $y$.
\begin{lemma}\label{lm:G}
The boundary-value problem (\ref{2.6}), (\ref{2.7}) has a unique solution which belongs to $W_2^2(\Om\setminus B_\d(x_0))\cap C^1(\overline{B_\d}\setminus\{x_0\})$ for all sufficiently small $\d>0$ and has the following asymptotic behavior in the vicinity of $x_0$,
\begin{equation}\label{3.0}
G(x)=\ln r+a+r\big((a_1\sin\tht+a_2\cos\tht)\ln r+P(\sin\tht,\cos\tht)\big) + O(r^2\ln^2r),\quad\, x\to x_0,
\end{equation}
where $a$ is a real number, $a_1,a_2\in\mathbb{C}$, and $P$ is a polynomial.
\end{lemma}
\begin{proof}
The differential expression \eqref{2.1} can be rewritten as
\begin{equation}\label{3.19}
\hat{\Op}=-\sum\limits_{i,j=1}^{2} A_{ij}\frac{\p^2\ }{\p x_i\p x_j} +\sum\limits_{j=1}^{2}\left(2\iu A_j-\sum\limits_{i=1}^{2} \frac{\p A_{ij}}{\p x_i}\right)\frac{\p\ }{\p x_j}+\left(\iu\sum\limits_{j=1}^{2}\frac{\p A_j}{\p x_j}+A_0\right).
\end{equation}
Using this representation and passing to the local variables $y$ in the vicinity of the point $x_0$ introduced in the proof of Lemma~\ref{lm3.0}, it is straightforward to confirm that there exists a function
\begin{equation}\label{3.19a}
\begin{aligned}
G_0(x)=&\ln r+r\big((a_1\sin\tht+a_2\cos\tht)\ln r+P(\sin\tht,\cos\tht)\big)
\\
&+r^2\big(P_1(\sin\tht,\cos\tht)\ln^2r+P_2(\sin\tht,\cos\tht)\ln r+P_3(\sin\tht,\cos\tht)\big),
\end{aligned}
\end{equation}
where $P$ and $P_i,\, =1,2,3,$ are some polynomials, such that the function $F_0(x):=(\hat{\Op}+c_1)G_0(x)$ is continuously differentiable in the punctured neighborhood of the point $x_0$ and exhibits there the following asymptotics,
\begin{equation}\label{3.19b}
F_0(x)=\mathcal{O}(r\ln r),\quad\; x\to x_0.
\end{equation}
We seek the solution to the boundary-value problem (\ref{2.6}), (\ref{2.7}) in the form
\begin{equation}\label{3.1}
G(x)=G_1(x)+G_2(x), \quad\; G_1:=\chi_\Om G_0,
\end{equation}
where for the unknown function $G_2$ we obtain the operator equation
\begin{equation}\label{3.2}
 (\Op_\Om+c_1) G_2=F,\quad\; F:= -\chi_\Om F_0+F_1.
\end{equation}
Here $F_1$ is a linear combination of the products of the derivatives of $G_0$ and $\chi_\Om$ up to the second order. If $\d>0$ is chosen small enough to ensure that $B_{2\d}(x_0)\subset \Om$, the above indicated properties of the function $F_0$ imply that $F$ belongs to $L_2(\Om)\cap C^\g(\overline{B_\d(x_0)})$ for all $\g\in(0,1)$.

Since the resolvent $(\Op_\Om+c_1)^{-1}$ is well-defined, equation (\ref{3.2}) has a unique solution which belongs to $\Dom(\Op_\Om)$. Moreover, using the standard Schauder estimates \cite{GT83}, we infer that it also belongs to $C^{2+\g}(\overline{B}_\d)$, which means, in particular, that the function $G_2$ has the Taylor expansion,
\begin{equation}\label{3.2a}
G_2(x)=a + a_3 y_1+ a_4 y_2+O(|y|^2),\quad\; x\to x_0,\;\; a_3,a_4\in\mathbb{C},
\end{equation}
where $y_1,y_2$ are the components of the vector $y=\mathrm{A}^{-\frac{1}{2}}(x-x_0)$. Returning to the function $G$, we conclude that problem (\ref{2.6}), (\ref{2.7}) is uniquely solvable and identity (\ref{3.0}) holds true.

It remains to check that the number $a$ is real. According (\ref{3.2}) and the definition of the functions $G_0$ and $F$ we have the identity
\begin{equation}\label{3.27}
\hf_\Om[G_2]+c_1\|G_2\|_{L_2(\Om)}^2=(F,G_2)_{L_2(\Om)} =-\big((\hat{\Op}+c_1)G_1,G_2\big)_{L_2(\Om)},
\end{equation}
which can be rewritten as follows,
\begin{equation}\label{3.23}
\hf_\Om[G_2]+c_1\|G_2\|_{L_2(\Om)}^2 -\big((\hat{\Op}+c_1)G_1,G_1\big)_{L_2(\Om)} =-\big((\hat{\Op}+c_1)G_1,G\big)_{L_2(\Om)}.
\end{equation}
We denote  $\Om^{\tilde{\d}}:=\Om\setminus \{x:\, |y|<\tilde{\d}\}$. In the last term on the left-hand side of (\ref{3.23}) we integrate by parts once bearing in mind the asymptotics (\ref{3.0}), (\ref{3.2a}), the identity (\ref{5.5}), and the fact that $G_1=G_0$ holds in the vicinity of the point $x_0$, obtaining
\begin{equation}\label{3.25}
\begin{aligned}
((\Op_\Om +c_1)G_1,G_1)_{L_2(\Om)}=&\lim\limits_{\tilde{\d}\to+0}\Bigg(
\sum\limits_{i,j=1}^{2} \left(A_{ij}\frac{\p G_1}{\p x_j}, \frac{\p G_1}{\p x_i}\right)_{L_2(\Om^{\tilde{\d}})} -2\IM\sum\limits_{j=1}^{2}\left(A_j\frac{\p G_1}{\p x_j},G_1\right)_{L_2(\Om^{\tilde{\d}})}
\\
&\hphantom{\lim\limits_{\tilde{\d}\to+0}\Bigg(}
+((A_0+c_1)G_1,G_1)_{L_2(\Om^{\tilde{\d}})}
-\int\limits_{\{x:\, |y|=\tilde{\d}\}} \overline{G_0}\,\frac{\p G_0}{\p\mathrm{n}}\,ds
\Bigg)
\\
=&\lim\limits_{\tilde{\d}\to+0}\Bigg(
\sum\limits_{i,j=1}^{2} \left(A_{ij}\frac{\p G_1}{\p x_j}, \frac{\p G_1}{\p x_i}\right)_{L_2(\Om^{\tilde{\d}})} -2\IM\sum\limits_{j=1}^{2}\left(A_j\frac{\p G_1}{\p x_j},G_1\right)_{L_2(\Om^{\tilde{\d}})}
\\
&\hphantom{\lim\limits_{\tilde{\d}\to+0}\Bigg(}
+((A_0+c_1)G_1,G_1)_{L_2(\Om^{\tilde{\d}})}
+\pi\tr \mathrm{A}\ln\tilde{\d}\Bigg).
\end{aligned}
\end{equation}
In the same way we integrate by parts twice on the right-hand side of (\ref{3.23}),
\begin{equation}\label{3.28}
\big((\hat{\Op}+c_1)G_1,G\big)_{L_2(\Om)}=\lim\limits_{\tilde{\d}\to+0} \int\limits_{\{x:\, |y|=\tilde{\d}\}}\left(G\overline{\frac{\p G_0}{\p\mathrm{n}}}   -\overline{G_0}\,\frac{\p G}{\p\mathrm{n}}\right)\,ds=-\pi a \tr\mathrm{A}.
\end{equation}
Substituting this identity together with (\ref{3.25}) into (\ref{3.23}), we obtain a formula for the constant $a$ showing that it is real. This concludes the proof.
\end{proof}

Denote next $\Pi_\e:=B_{2R_2}(x_0)\setminus\om_\e$, then we have the following result.

\begin{lemma}\label{lm:bnd}
For all $v\in\H^1(\Pi_\e)$ the estimate
\begin{equation}\label{3.3}
\|v\|_{L_2(\p\om_\e)}^2\leqslant C\e\Big(|\ln\e|\|\nabla v\|_{L_2(\Pi_\e)}^2+\|v\|_{L_2(\Pi_\e)}^2\Big)
\end{equation}
is valid, where $C$ is a fixed constant independent of $\e$ and $v$. If, in addition, the function $v$ is defined on entire ball $B_{2R_2}(x_0)$ and belongs to $\H^1(B_{2R_2}(x_0))$, then the estimate
\begin{equation}\label{3.3a}
\|v\|_{L_2(B_{2R_2}(x_0)}^2\leqslant C\e^2\left( |\ln\e|  \|\nabla v\|_{L_2(B_{R_2}(x_0)}^2 +  \|v\|_{L_2(B_{2R_2}(x_0)}^2\right),
\end{equation}
holds, where $C$ is a fixed constant independent of $\e$ and $v$.
\end{lemma}
\begin{proof}
We denote by $\chi:\: \mathbb{R}_+\to [0,1]$ an infinitely differentiable cut-off function, equal to one if $t<1$ and vanishing for $t>2$. It is clear that
\begin{equation}\label{3.4}
v(x)=v(x)\chi\left(\frac{|x-x_0|}{R_1\e}\right)=:v_\e\;\;\text{on}\;\; \p\om_\e,\quad\text{and}\;\; v_\e=0\;\;\text{on}\;\; \p B_{2R_1\e}(x_0).
\end{equation}
We rescale variables, $x\mapsto (x-x_0)\e^{-1}$, and by standard embedding theorems we get
\begin{equation}\label{3.5}
\begin{aligned}
\|v\|_{L_2(\p\om_\e)}^2=&\:\e\|v_\e(x_0+\e\,\cdot\,)\|_{L_2(\p\om)}^2
\\[.3em]
\leqslant&\: C\e\|\nabla v_\e(x_0+\e\,\cdot\,)\|_{L_2(B_{2R_1}(0)\setminus\om)}^2
\\[.3em]
=&\:C\e\|\nabla v_\e\|_{L_2(B_{2R_1\e}(x_0)\setminus\om_\e)}^2
\\[-.1em]
\leqslant &\: C\left(\e\|\nabla v \|_{L_2(B_{2R_1\e}(x_0)\setminus\om_\e)}^2+ \e^{-1}\|v\|_{L_2(B_{2R_1\e}(x_0)\setminus B_{R_1\e}(x_0))}^2 \right),
\end{aligned}
\end{equation}
where the symbol $C$  stands for various inessential constants independent of $\e$ and $v$. Let us estimate the term $\|u\|_{L_2(B_{2R_1\e}(x_0)\setminus B_{R_1\e}(x_0))}^2$.

It follows from (\ref{3.6}) that
\begin{equation*}
v(x)=v(x)\chi\left(\frac{|x-x_0|}{R_2}\right)\quad\text{in}\;\; B_{2R_1\e}(x_0)\setminus \om_\e.
\end{equation*}
We denote $r:=|x-x_0|$ for $x\in B_{2R_1\e}(x_0)\setminus B_{R_1\e}(x_0)$, and furthermore, we put $x':= x_0+\frac{r'}{r}(x-x_0)$, then we have
\begin{equation*}
|v(x)|=\left|\,\int\limits_{2R_2}^{r} \frac{\p\ }{\p r}\left(v(x')\chi\left(\frac{r'}{R_2}\right)\right)\,\mathrm{d}r'\right| \leqslant \int\limits_{r}^{2R_2} |\nabla v(x')|\chi\left(\frac{r'}{R_2}\right) \,\mathrm{d}r' + \frac{1}{R_2} \int\limits_{r}^{2R_2} |v(x')|\left|\chi'\left(\frac{r'}{R_2}\right) \right| \,\mathrm{d}r'.
\end{equation*}
Using next Cauchy-Schwarz inequality together with the properties of the cut-off function, we arrive at the estimate
\begin{equation}\label{3.8a}
\begin{aligned}
|v(x)|^2\leqslant&\: 2\left(\int\limits_{r}^{2R_2} |\nabla v(x')|\chi\left(\frac{r'}{R_2}\right) \,\mathrm{d}r'\right)^2 + \frac{2}{R_2^2} \left(\int\limits_{r}^{2R_2} |v(x')|\left|\chi'\left(\frac{r'}{R_2}\right) \right| \,\mathrm{d}r'\right)^2
\\
\leqslant &\: 2\ln\frac{2R_2}{r} \int\limits_{r}^{2R_2} |\nabla v(x')|^2 r' \,\mathrm{d}r' + \frac{2\ln 2R_2}{R_2^2} \Big(\sup\limits_{t\in[1,2]} |\chi'(t)|\Big)^2\int\limits_{r}^{2R_2} |v(x')|^2r' \,\mathrm{d}r'
\end{aligned}
\end{equation}
Integrating this inequality over $B_{2R_1\e}(x_0)\setminus B_{R_1\e}(x_0)$ we find
\begin{equation*}
\|v\|_{L_2(B_{2R_1\e}(x_0)\setminus B_{R_1\e}(x_0))}^2\leqslant C\e^2\left( |\ln\e|  \|\nabla v\|_{L_2(B_{R_2}(x_0)\setminus B_{R_1\e}(x_0))}^2 +  \|v\|_{L_2(B_{2R_2}(x_0)\setminus B_{R_2}(x_0))}^2\right),
\end{equation*}
and substituting finally from here into the right-hand side of (\ref{3.5}) we obtain the sought estimate (\ref{3.3}).  Finally, if 
$v\in\H^1(B_{2R_2(x_0)})$, we  integrate  estimate (\ref{3.8a}) over $B_{2R_1\e}(x_0)$ and arrive immediately at estimate (\ref{3.3a}) which concludes the proof.
\end{proof}

\begin{lemma}\label{lm:bnd-mean}
For all $v\in \H^1(\Pi_\e)$ satisfying the condition
\begin{equation}\label{3.10}
\int\limits_{\p\om_\e} v\,ds=0
\end{equation}
the inequality
\begin{equation}\label{3.11}
\|v\|_{L_2(\p\om_\e)}^2\leqslant C\e\|\nabla v\|_{L_2(\Pi_\e)}^2
\end{equation}
holds, where $C$ is a constant independent of $\e$ and $v$. If, in addition, the function $v$ is defined on the entire ball $B_{2R_2}(x_0)$ and $v\in \H^2(B_{2R_2}(x_0))$, then
\begin{equation}\label{3.12}
\|v\|_{L_2(\p\om_\e)}^2\leqslant C\e^3|\ln\e|\|v\|_{\H^2(B_{2R_2}(x_0))}^2,
\end{equation}
where $C$ is a constant independent of $\e$ and $v$.
\end{lemma}
\begin{proof}
Throughout the proof the symbol $C$ stands for various inessential constants independent of $\e$ and $v$. The function
\begin{equation}\label{3.15}
v_\bot:=v-\la v\ra_\om,\quad\; \la v\ra_\om:=\frac{1}{\e^2|B_{2R_1(0)}\setminus\om|} \int\limits_{B_{2R_1\e}(x_0)\setminus\om_\e} v\,\mathrm{d}x,
\end{equation}
obviously satisfies the identities
\begin{equation}\label{3.16}
\int\limits_{B_{2R_1\e}(x_0)\setminus\om_\e} v_\bot(x)\, \mathrm{d}x=0,\qquad
\int\limits_{B_{2R_1}(0)\setminus\om} v_\bot(x_0+\e\,\cdot\,)\, \mathrm{d}x=0,
\end{equation}
which allow us to apply the Poincar\'e inequality in the following chain of estimates,
\begin{align*}
\|v_\bot\|_{L_2(\p\om_\e)}^2=\e
\|v_\bot(x_0+\e\,\cdot\,)\|_{L_2(\p\om)}^2\leqslant C\e
\|\nabla v_\bot(x_0+\e\,\cdot\,)\|_{L_2(B_{2R_1}(0)\setminus\om
)}^2 \leqslant C\e
\|\nabla v\|_{L_2(B_{2R_1\e}(x_0)\setminus\om_\e)}^2.
\end{align*}
From this inequality in combination with (\ref{3.15}) we infer that
\begin{equation}\label{3.17}
\|v\|_{L_2(\p\om_\e)}^2=\|v_\bot+\la v\ra_\om\|_{L_2(\p\om_\e)}^2 \leqslant C\e \Big(
\|\nabla v\|_{L_2(B_{2R_1\e}(x_0)\setminus\om_\e)}^2 +  |\la v\ra_\om|^2
\Big).
\end{equation}
Let us assess $\la v\ra_\om$. In the domain $B_{2R_1(0)}\setminus \om$ we consider the boundary-value problem
\begin{equation}\label{3.13}
\begin{gathered}
\D X=\frac{|\p\om|}{|B_{2R_1}(0)\setminus\om|}\quad\text{in}\;\; B_{2R_1(0)}\setminus\om,
 \\
 \frac{\p X}{\p\nu}=1\quad\text{on}\;\; \p\om,\quad \frac{\p X}{\p\nu}=0\quad\text{on}\;\; \p B_{2R_1}(0),
\end{gathered}
\end{equation}
where $\nu$ is the unit outward normal to the boundary of $B_{2R_1}(0)\setminus\om$. This problem is solvable because we have
\begin{equation*}
|\p\om|=\int\limits_{\p\om} \,\mathrm{d}s= \int\limits_{\p\om} \frac{\p X}{\p\nu}\,\mathrm{d}s=\!\!\int\limits_{B_{2R_1}(0)\setminus\om}\!\! \D X\,\mathrm{d}x= |B_{2R_1}(0)\setminus\om|.
\end{equation*}
In view of the assumed smoothness of the boundary $\p\om$ and the standard Schauder estimate, we can conclude that $X\in C^{(2+\g)}(\overline{B_{2R_1}(0)\setminus\om})$ for all $\g\in(0,1)$. A solution to problem (\ref{3.13}) is defined up to an additive constant which we fix it by the requirement
\begin{equation}\label{3.14}
\int\limits_{B_{2R_1}(0)\setminus\om} X(x)\,\mathrm{d}x=0.
\end{equation}
Combining problem (\ref{3.13}) and assumption (\ref{3.10}), we can rewrite $\la v\ra_\om$ using integration by parts,
\begin{equation}\label{3.9}
\begin{aligned}
\la v\ra_\om=&\:\frac{1}{|\p\om|}\int\limits_{B_{2R_1\e(x_0)}\setminus\om_\e} v(x)\D X \left(\frac{x-x_0}{\e}\right)\,\mathrm{d}x
\\
=&\: \frac{1}{\e|\p\om|} \int\limits_{\p\om_\e} v(x)\,\mathrm{d}s - \frac{1}{\e|\p\om|} \int\limits_{B_{2R_1\e}(x_0)\setminus\om_\e} \nabla v(x)\cdot(\nabla X)\left(\frac{x-x_0}{\e}\right)\,\mathrm{d}x
\\
=&\:- \frac{1}{\e|\p\om|} \int\limits_{B_{2R_1\e}(x_0)\setminus\om_\e} \nabla v(x)\cdot(\nabla X)\left(\frac{x-x_0}{\e}\right)\,\mathrm{d}x,
\end{aligned}
\end{equation}
and consequently, by Cauchy-Schwarz inequality we can infer that
\begin{equation}\label{3.18}
\begin{aligned}
|\la v\ra_\om|^2\leqslant &\:\frac{1}{\e^2|\p\om|^2} \|\nabla v\|_{L_2(B_{2R_1\e}(x_0)\setminus\om_\e)}^2 \left\|(\nabla X)\left(\frac{x-x_0}{\e}\right)\right\|_{L_2(B_{2R_1\e}(x_0)\setminus\om_\e)}^2
\\[.3em]
\leqslant &\: C\|\nabla v\|_{L_2(B_{2R_1\e}(x_0)\setminus\om_\e)}^2.
\end{aligned}
\end{equation}
This estimate together with (\ref{3.17}) yields inequality (\ref{3.11}).

Assume finally that $v\in\H^2(B_{2R_2}(x_0))$. Then we can replace $v$ in (\ref{3.8a}) with $\frac{\p v}{\p x_i}$, $i=1,2$, and integrate such an estimate over $B_{2R_1\e}(x_0)$. This gives
\begin{equation}\label{3.24}
\|\nabla v\|_{L_2(2B_{R_1\e}(x_0))}^2\leqslant
\|\nabla v\|_{L_2(2B_{R_1\e}(x_0))}^2 \leqslant C\e^2|\ln\e| \|v\|_{\H^2(B_{2R_2}(x_0))},
\end{equation}
which in combination with (\ref{3.17}), (\ref{3.18}) implies (\ref{3.12}) concluding thus the proof.
\end{proof}

Next we consider for any $v\in \H^1(\Pi_\e)$ the mean value over the boundary of $\om_\e$,
 \begin{equation}\label{5.1a}
\la v\ra_{\p\om_\e}:=\frac{1}{\e|\p\om|}\int\limits_{\p\om_\e} v\,ds.
\end{equation}

\begin{lemma}\label{lm4.0}
For all $\vp\in C(\p\om)$ and all $v\in \H^2(B_{2R_2}(x_0))$ the inequality
\begin{equation}\label{3.20}
\bigg|\e^{-1}\int\limits_{\p\om_\e} \vp\left(\frac{s_\e}{\e}\right)v(x)\,\mathrm{d}s- c(\vp) v(x_0)\bigg|\leqslant C\e |\ln\e|^\frac{1}{2}\|v\|_{\H^2(B_{2R_2}(x_0))},\quad c(\vp):=\int\limits_{\p\om} \vp(s)\,\mathrm{d}s,
\end{equation}
holds true, where $C$  is a constant  independent of $\e$ and $v$.
\end{lemma}
\begin{proof}
We put
\begin{equation*}
v^\bot:=v-\la v\ra_{\p\om_\e},\quad\; \int\limits_{\p\om_\e} v_\bot\, \mathrm{d}s=0,
\end{equation*}
and note the following obvious identity,
\begin{equation}\label{3.30a}
\begin{aligned}
\e^{-1}\int\limits_{\p\om_\e} \vp\left(\frac{s_\e}{\e}\right)v(x)\, \mathrm{d}s=&\:\e^{-1}\la v\ra_{\p\om_\e}\int\limits_{\p\om_\e} \vp\left(\frac{s_\e}{\e}\right)\,\mathrm{d}s+\e^{-1}\int\limits_{\p\om_\e} \vp\left(\frac{s_\e}{\e}\right)v^\bot(x)\,\mathrm{d}s
\\
=&\:c(\vp)\la v\ra_{\p\om_\e} +\e^{-1}(v^\bot,\vp)_{L_2(\p\om_\e)}
\end{aligned}
\end{equation}
and from Lemma~\ref{lm:bnd-mean} we get
\begin{equation}\label{3.30b}
\Big|\e^{-1}(v^\bot,\vp)_{L_2(\p\om_\e)}\Big|\leqslant  C\e|\ln\e|^\frac{1}{2}\|v\|_{\H^2(B_{2R_2}(x_0))}.
\end{equation}
Let us assess the difference $\la v\ra_{\p\om_\e}-v(x_0)$. To this aim, we consider the boundary-value problem
\begin{equation}\label{3.21}
\D Y=0
\quad\text{in}\;\; \om\setminus\{0\},
 \quad\;
 \frac{\p Y}{\p\nu}=1\quad\text{on}\;\; \p\om,\quad\; Y(x)=\frac{|\p\om|}{2\pi}\ln|x-x_0|+O(1),\quad x\to x_0,
\end{equation}
where $\nu$ is the unit outward normal to the boundary of $\om$. This problem has a unique solution up to a constant which can be chose in such a way that
\begin{equation}\label{3.22}
\int\limits_{\om} Y(x)\,\mathrm{d}x=0.
\end{equation}
in view of the assumed smoothness of the boundary $\p\om$ and the standard Schauder estimate,  we have $Y\in C^{(2+\g)}(\overline{\om\setminus B_\d(0)})$ for any $\g\in(0,1)$ and all $\d>0$.

Let $v\in C^2(\om_\e)$. Using integration by parts and taking into account the indicated properties of the function $Y$ we get
\begin{equation*}
0=\int\limits_{\om_\e} v\D Y\left(\frac{x-x_0}{\e}\right)\,\mathrm{d}x=\e^{-1}\int\limits_{\p\om_\e} v(x)\,\mathrm{d}s-\int\limits_{\p\om_\e} Y\left(\frac{x-x_0}{\e}\right)\frac{\p v}{\p\nu}(x)\,\mathrm{d}s-|\p\om|v(x_0).
\end{equation*}
Since the space $C^2(\overline{\om_\e})$ is dense in $\H^2(\om_\e)$, the above identity holds for all $v\in \H^2(\om_\e)$ as well, and by Cauchy-Schwarz inequality and Lemma~\ref{lm:bnd} it implies
\begin{align*}
\big|\la v\ra_{\p\om_\e}-v(x_0)\big|=&\:\frac{1}{|\p\om|}\bigg|
\int\limits_{\p\om_\e} Y\left(\frac{x-x_0}{\e}\right)\frac{\p v}{\p\nu}(x)\,\mathrm{d}s
\bigg|
\\[.3em]
\leqslant&\: C\e^\frac{1}{2}\left\|\nabla v\right\|_{L_2(\p\om_\e)}\leqslant C\e|\ln\e|^\frac{1}{2}\|v\|_{\H^2(\om_\e)}.
\end{align*}
which together with (\ref{3.30a}), (\ref{3.30b}) yields the sought result.
\end{proof}

\section{
Convergence}\label{ss:proof}

The goal of this section is to prove Theorems~\ref{th:main}  and~\ref{th2.2}. The argument consists of two main parts. In the first we establish the self-adjointness of the operators $\Op_\e$ and $\Op_{0,\b}$, while the second part is devoted to the verification of the norm resolvent convergence and the spectral convergence.

\subsection{Self-adjointness of the operator $\Op_\e$}

We start by introducing a sesquilinear form  $\hf_\e$ in $L_2(\Om_\e)$ defined by the identity
\begin{equation}\label{5.11}
\begin{aligned}
\hf_\e(u,v):=&\:\hf_\Om\big((1-\chi_\Om)u,(1-\chi_\Om)v\big) + \hf_{\Om_0\setminus\om_\e}\big(\chi_\Om u,(1-\chi_\Om)v\big)
 \\
&\:+ \hf_{\Om_0\setminus\om_\e}\big((1-\chi_\Om)u,\chi_\Om v\big)
+ \hf_{\Om_0\setminus\om_\e}(\chi_\Om u,\chi_\Om v)-\e^{-1}(\a u,v)_{L_2(\p\om_\e)}
\end{aligned}
\end{equation}
on the domain
\begin{equation}\label{5.21}
\Dom(\hf_\e):=\Big\{u:\, (1-\chi_\Om)u\in\Dom(\hf_\Om),\ \chi_\Om u\in\H^1(\Om_0\setminus\om_\e)\Big\}
\end{equation}
It is clear that  this form is symmetric; let us check that it is associated with the operator $\Op_\e$, in other words, that we have
\begin{equation}\label{5.12}
\hf_\e(u,v)=(\Op_\e u,v)_{L_2(\Om_\e)}\quad\;\text{for all}\;\;  u\in\Dom(\Op_\e), \ v\in\Dom(\hf_\e).
\end{equation}
Indeed, since $u\in\H^2(\Om_0\setminus\om_\e)$, $v\in\H^1(\Om_0\setminus\om_\e)$, according the definition of $\Op_\e$, $\hf$ and $\chi_\Om$, we can use integration by parts to rewrite the last four terms on the right-hand side of \eqref{5.11} as follows,
\begin{equation}\label{5.13}
\begin{aligned}
 \hf_{\Om_0\setminus\om_\e}\big(&\chi_\Om u,(1-\chi_\Om)v\big)
+  \hf_{\Om_0\setminus\om_\e}\big((1-\chi_\Om)u,\chi_\Om v\big)
+  \hf_{\Om_0\setminus\om_\e}(\chi_\Om u,\chi_\Om v)-\e^{-1}(\a u,v)_{L_2(\p\om_\e)}
\\
&=\big(\Op_\Om\chi_\Om u,(1-\chi_\Om)v\big)_{L_2(\Om_0\setminus\om_\e)}
+  \big(\Op_\Om(1-\chi_\Om)u,\chi_\Om v\big)_{L_2(\Om_0\setminus\om_\e)}
+  (\Op_\Om\chi_\Om u,\chi_\Om v)_{L_2(\Om_0\setminus\om_\e)}
\\
&=\big(\Op_\Om\chi_\Om u,v\big)_{L_2(\Om_\e)}
+  \big(\Op_\Om(1-\chi_\Om)u,\chi_\Om v\big)_{L_2(\Om_\e)}.
\end{aligned}
\end{equation}
As for the remaining term, since by assumptions made about the cut-off function $\chi_\Om$ we have $(1-\chi_\Om)u\in\Dom(\Op_\Om)$ and $(1-\chi_\Om)v\in\Dom(\hf_\Om)$, we infer that
\begin{equation}\label{5.14}
\hf_\Om\big((1-\chi_\Om)u,(1-\chi_\Om)v\big)=\big(\Op_\Om(1-\chi_\Om)u,(1-\chi_\Om)v\big)_{L_2(\Om)}
=\big(\Op_\Om(1-\chi_\Om)u,(1-\chi_\Om)v\big)_{L_2(\Om_\e)
}.
\end{equation}
We also have $(1-\chi_\Om)u\in\H^2(\Om_0)$, and therefore
\begin{equation}\label{5.15}
\Op_\Om(1-\chi_\Om)u=\hat{\Op} (1-\chi_\Om)u\quad\text{on}\;\; \Om\setminus B_{R_2}(x_0).
\end{equation}
Substituting this identity together with (\ref{5.14}), (\ref{5.13}) into  definition (\ref{5.11}) we get
\begin{equation}\label{5.16}
\begin{aligned}
\hf_\e(u,v)=&\:\big(\Op_\Om(1-\chi_\Om)u,(1-\chi_\Om)v\big)_{L_2(\Om_\e)}
 + \big(\hat{\Op}\chi_\Om u,v\big)_{L_2(\Om_0\setminus\om_\e)}
\\
&+  \big(\hat{\Op}(1-\chi_\Om)u,\chi_\Om v\big)_{L_2(\Om_0\setminus\om_\e)}
\\
=&\:\big(\Op_\Om(1-\chi_\Om)u, v\big)_{L_2(\Om_\e)} -\big(\Op_\Om(1-\chi_\Om)u,\chi_\Om)v\big)_{L_2(\Om_\e)}
\\
& + \big(\hat{\Op}\chi_\Om u,v\big)_{L_2(\Om_0\setminus\om_\e)}
+  \big(\hat{\Op}(1-\chi_\Om)u,\chi_\Om v\big)_{L_2(\Om_0\setminus\om_\e)}
\\
=&\:\big(\Op_\Om(1-\chi_\Om)u, v\big)_{L_2(\Om_\e)} + \big(\hat{\Op}\chi_\Om u,v\big)_{L_2(\Om_\e)}
\end{aligned}
\end{equation}
which proves relation (\ref{5.12}).

\medskip

Our next step is to check that the form $\hf_\e$ is semibounded from below. Here we shall make use of the following two auxiliary results concerning the function $G$ introduced in Lemma~\ref{lm:G}.

\begin{lemma}\label{lm:(G,u)}
For all $u\in\Dom(\hf_\e)$ we have the identity
\begin{equation}\label{5.17}
\hf_\e(G,u)+c_1(G,u)_{L_2(\Om_\e)}=\Big(\frac{\p G}{\p\mathrm{n}} -\e^{-1}\a G,u\Big)_{L_2(\p\om_\e)}.
\end{equation}
\end{lemma}
\begin{proof}
To begin with, we observe that $G\in\Dom(\Op_\e)$, and therefore the quantity $\hf_\e(u,G)$ is well defined. It also follows from the definition of the function $G$ and (\ref{5.15}) that
\begin{equation}\label{5.18}
\hf_\Om((1-\chi_\Om)G,(1-\chi_\Om)u)+c_1\big( G,(1-\chi_\Om)u\big)_{L_2(\Om_\e)} =- \hf_{\Om_0\setminus\om_\e}(\chi_\Om G,(1-\chi_\Om)u).
\end{equation}
Thus we have
\begin{align*}
\hf_\e(G,u)+c_1(G,u)_{L_2(\Om_\e)}=&\: \hf_{\Om_0\setminus\om_\e}\big((1-\chi_\Om)G,\chi_\Om u\big)
+  \hf_{\Om_0\setminus\om_\e}(\chi_\Om G,\chi_\Om u)-\e^{-1}(\a G,u)_{L_2(\p\om_\e)}
\\
&\:+c_1\big( G, \chi_\Om u\big)_{L_2(\Om_\e)}
\\
=&\: \hf_{\Om_0\setminus\om_\e}(G,\chi_\Om u)+c_1(G,\chi_\Om u)_{L_2(\Om_\e)}-\e^{-1}(\a G,u)_{L_2(\p\om_\e)}\,;
\end{align*}
integrating then by parts and using the equation that $G$ satisfies, we arrive at (\ref{5.17}).
\end{proof}

\begin{lemma}\label{lm:asG}
The identities
\begin{align}
& \label{5.8}
\left(\frac{\p G}{\p\mathrm{n}} -\e^{-1}\a G,G\right)_{L_2(\p\om_\e)}= K+\pi a\tr\mathrm{A}+\mathcal{O}(\ln^{-1}\e),
\\
\label{5.2}
&
\begin{aligned}
\left(\frac{\p G}{\p\mathrm{n}} -\e^{-1}\a G\right)(\e^{-1}s_\e) =&\:\e^{-1}\big(
\ln^{-1}\e\, \Phi_1(\e^{-1}s_\e) \\
&\:+\ln^{-2}\e\,\Phi_2(\e^{-1}s_\e,\e)\big)\quad\text{on}\;\;\p\om_\e
\end{aligned}
\end{align}
hold true, where where $ K$ is defined by formula (\ref{2.4a}) and
\begin{equation}\label{5.3}
\Phi_1(s):=-\a_0(s)(\ln|\mathrm{A}^{-\frac{1}{2}}\x(s)|+a)-\a_1(s),
\end{equation}

and $\Phi_2=\Phi_2(s,\e)$ is a function uniformly bounded in $\e$ and $s$.
\end{lemma}
\begin{proof}
The stated asymptotic expansions (\ref{5.2}), (\ref{5.3}) can be easily confirmed by straightforward calculations using relations (\ref{3.0}). These three formul{\ae} in combination with the identities (\ref{5.5}), (\ref{2.4}) yield in turn (\ref{5.8}) which concludes the proof.
\end{proof}

Given an arbitrary $u\in \H^1(\Om_\e)$ we denote
\begin{equation}\label{5.1b}
u^\bot:=u- \la u\ra_G G,\quad\; \la u\ra_G:=\frac{\la u\ra_{\p\om_\e}}{\la G\ra_{\p\om_\e}},
\end{equation}
recalling that the averaging $\la\cdot\ra_{\p\om_\e}$ was introduced in (\ref{5.1a}).  Then in view of Lemma~\ref{lm:(G,u)} we have
\begin{equation}\label{5.19}
\begin{aligned}
\hf_\e(u,u)+c_1\|u\|_{L_2(\Om_\e)}^2=&\hf_\e\big(u^\bot+\la u\ra_G G,u\big)+c_1\big(u^\bot+\la u\ra_G G,u\big)_{L_2(\Om_\e)}
\\
=&\:\hf_\e\big(u^\bot,u\big)+c_1(u^\bot,u)_{L_2(\Om_\e)} + \la u\ra_G \left(\frac{\p G}{\p\mathrm{n}} -\e^{-1}\a G,u\right)_{L_2(\p\om_\e)}
\\
=&\:\hf_\e\big[u^\bot]+c_1\|u^\bot\|_{L_2(\Om_\e)}^2 +|\la u\ra_G|^2 \left(\frac{\p G}{\p\mathrm{n}} -\e^{-1}\a G, G \right)_{L_2(\p\om_\e)}
 \\
 &\:+2\RE \la u\ra_G \left(\frac{\p G}{\p\mathrm{n}} -\e^{-1}\a G,u^\bot  \right)_{L_2(\p\om_\e)}.
\end{aligned}
\end{equation}
In view of the asymptotics (\ref{5.2})
we have
\begin{equation*}
\left(\frac{\p G}{\p\mathrm{n}} -\e^{-1}\a G,u^\bot\right)_{L_2(\p\om_\e)} =\e^{-1}\ln^{-1}\e\left(\Phi_1+\ln^{-1}\Phi_2,u^\bot    \right)_{L_2(\p\om_\e)}
\end{equation*}
and therefore, by virtue of Lemma~\ref{lm:bnd-mean},
\begin{equation}\label{5.20}
\left|\left(\frac{\p G}{\p\mathrm{n}} -\e^{-1}\a G,u^\bot\right)_{L_2(\p\om_\e)}\right|\leqslant C|\ln\e|^{-1}\|u^\bot\|_{\H^1(B_{2R_2}(x_0)\setminus\om_\e)},
\end{equation}
where $C$ is a constant independent of $\e$ and $u$.

To proceed we have to analyze the term $\hf_\e(u^\bot,u^\bot)$ in (\ref{5.19}). The estimate (\ref{2.0a}) implies
\begin{equation}\label{5.24}
\begin{aligned}
\hf_\Om\big[(1-\chi_\Om)u^\bot] + c_1\|u^\bot\|_{L_2(\Om\setminus\Om_0)}^2 \geqslant&\: c_2\|u^\bot\|_{\H^1(\Om\setminus\Om_0)}^2 + \hf_{\Om_0}\big[(1-\chi_\Om)u^\bot\big]
\\[.3em]
=&\:c_2\|u^\bot\|_{\H^1(\Om\setminus\Om_0)}^2 + \hf_{\Om_0\setminus\om_\e}\big[(1-\chi_\Om)u^\bot\big].
\end{aligned}
\end{equation}
At the same time, it is straightforward to confirm
that
\begin{equation}\label{5.24a}
\begin{aligned}
\hf_{\Om_0\setminus\om_\e}\big[(1-\chi_\Om)u^\bot\big] &+ \hf_{\Om_0\setminus\om_\e}\big((1-\chi_\Om)u^\bot,\chi_\Om u^\bot
\big) \\
&+ \hf_{\Om_0\setminus\om_\e}\big(\chi_\Om u^\bot,(1-\chi_\Om)u^\bot\big)
+\hf_{\Om_0\setminus\om_\e}[ \chi_\Om u^\bot]=\hf_{\Om_0\setminus\om_\e}[u^\bot],
\end{aligned}
\end{equation}
hence by the definition of the form $\hf_\e$ and estimates (\ref{5.24}), (\ref{2.0c}) we get
\begin{equation}\label{5.22}
\begin{aligned}
\hf_\e[u^\bot] + c_1\|u^\bot\|_{L_2(\Om_\e)}^2 \geqslant &\: c_2\|u^\bot\|_{\H^1(\Om\setminus\Om_0)} + \hf_{\Om_0\setminus\om_\e}[u^\bot]
 \\[.3em]
 &\:+ c_1\|u^\bot\|_{L_2(\Om_0\setminus\om_\e)}-\e^{-1}(\a u^\bot,u^\bot)_{L_2(\p\om_\e)}
 \\[.3em]
 \geqslant &\: c_2 \|u^\bot\|_{\H^1(\Om_\e)}-\e^{-1}(\a u^\bot,u^\bot)_{L_2(\p\om_\e)}.
\end{aligned}
\end{equation}
By Lemma~\ref{lm:bnd-mean} and the definition of $\a$ by \eqref{2.2c} and \eqref{2.4} we also have
\begin{equation}\label{5.22a}
\big|\e^{-1}(\a u^\bot,u^\bot)_{L_2(\p\om_\e)}\big|\leqslant C|\ln\e|^{-1} \|\nabla u^\bot\|_{L_2(\Om_0\setminus\om_\e)}^2,
\end{equation}
where $C$ is a fixed constant independent of $\e$ and $u$, hence in view of (\ref{5.22}) we finally obtain
\begin{equation}\label{5.23}
\hf_\e[u^\bot] + c_1\|u^\bot\|_{L_2(\Om_\e)}^2\geqslant \big(c_2-C|\ln\e|^{-1}\big) \|u^\bot\|_{L_2(\Om_\e)}^2.
\end{equation}
This estimate together with (\ref{5.20}), (\ref{5.19}), and (\ref{5.8}) implies that
\begin{equation}\label{5.25}
\hf_\e(u,u)+c_1\|u\|_{L_2(\Om_\e)}^2 \geqslant (c_2-C|\ln\e|^{-1} ) \|u^\bot\|_{\H^1(\Om_\e)}^2+ (K+\pi a\tr\mathrm{A} -C|\ln\e|^{-1}) |\la u\ra_G|^2,
\end{equation}
where $C$ is again a fixed constant independent of $\e$ and $u$. Furthermore, using Cauchy-Schwarz inequality it is easy to check that
\begin{equation}\label{5.26}
\begin{aligned}
\|u\|_{L_2(\Om_\e)}^2=&\:\|u^\bot\|_{L_2(\Om_\e)}^2+2\RE\la u\ra_G(G,u^\bot)_{L_2(\Om_\e)}+|\la u\ra_G|^2\|G\|_{L_2(\Om_\e)}^2
\\[.3em]
\geqslant &\:-\eta\|u^\bot\|_{L_2(\Om_\e)}^2+\frac{\eta\|G\|_{L_2(\Om_\e)}^2}{1+\eta}|\la u\ra_G|^2
\end{aligned}
\end{equation}
holds for an arbitrary $\eta\in(0,1)$, and this identity in turn implies
\begin{equation}\label{5.27}
\begin{aligned}
(c_2&-C|\ln\e|^{-1} ) \|u^\bot\|_{\H^1(\Om_\e)}^2 + (K+\pi a\tr\mathrm{A} -C|\ln\e|^{-1}) |\la u\ra_G|^2+c_3\|u\|_{L_2(\Om_\e)}^2 \\[.3em]
&\: \geqslant (c_2-C|\ln\e|^{-1}) \|\nabla u^\bot\|_{L_2(\Om_\e)}^2
+ (c_2-C|\ln\e|^{-1} -c_3\eta)\| u^\bot\|_{L_2(\Om_\e)}^2 \\[.3em] &\: + \bigg(K+\pi a\tr\mathrm{A}-C|\ln\e|^{-1} + c_3\frac{\eta\|G\|_{L_2(\Om_\e)}^2}{1+\eta}
\bigg)|\la u\ra_G|^2
\end{aligned}
\end{equation}
for any $c_3>0$. Having in mind that $\|G\|_{L_2(\Om_\e)}^2=\|G\|_{L_2(\Om)}^2+o(1)$, we choose $c_3$ and $\eta$ in such a way that $c_3\eta$ is less than $c_2$ and $\eta$ is small enough. In view of (\ref{2.4a}) we can achieve that
\begin{equation}\label{5.27a}
c_2-C|\ln\e|^{-1} -c_3\eta\geqslant c_4,\qquad  K+\pi a\tr\mathrm{A}-C|\ln\e|^{-1} + c_3\frac{\eta\|G\|_{L_2(\Om_\e)}^2}{1+\eta}\geqslant c_4
\end{equation}
holds for all sufficiently small $\e$, where $c_4$ is a fixed positive constant independent of $\e$, and $c_3$ is independent of $\e$ as well. By means of (\ref{5.25}), (\ref{5.27}) we then have
\begin{equation}\label{5.28}
\hf_\e[u]+c_5\|u\|_{L_2(\Om_\e)}^2 \geqslant c_6 \big( \|u^\bot\|_{\H^1(\Om_\e)}^2+   |\la u\ra_G|^2\big),
\end{equation}
where $c_5$ and $c_6$ are fixed constants independent of $\e$ and $u$.

We also observe that we if we replace estimate (\ref{5.24}) by the identity
\begin{align*}
\hf_\Om\big[(1-\chi_\Om)u^\bot] + c_1\|u^\bot\|_{L_2(\Om\setminus\Om_0)}^2=&\: \hf_\Om\big[(1-\chi_\Om)u^\bot] + c_1\|u^\bot\|_{L_2(\Om\setminus\Om_0)}^2
\\
&\:- \hf_{\Om_0}\big[(1-\chi_\Om)u^\bot\big]+\hf_{\Om_0}\big[(1-\chi_\Om)u^\bot\big]
\end{align*}
and proceed as in (\ref{5.24a})--(\ref{5.27a}), 
 taking in addition (\ref{5.28}) into account, we get one more estimate, namely
\begin{equation}\label{5.28a}
\begin{aligned}
\hf_\e[u]+c_5\|u\|_{L_2(\Om_\e)}^2 \geqslant &\: \hf_\Om\big[(1-\chi_\Om)u^\bot] + c_1\|u^\bot\|_{L_2(\Om\setminus\Om_0)}^2- \hf_{\Om_0}\big[(1-\chi_\Om)u^\bot\big]
\\
&\:
+c_6 \big( \|u^\bot\|_{\H^1(\Om_\e)}^2+   |\la u\ra_G|^2\big),
\end{aligned}
\end{equation}

Finally, let us show that the form $\hf_\e$ is closed. We recall that the domain $\Dom(\hf_\Om)$ is by assumption a subspace in $\H^1(\Om)$ and take an arbitrary sequence $u_n\in\Dom(\hf_\e)$ such that
\begin{equation}\label{5.29}
\|u_n-u\|_{L_2(\Om_\e)}\to0,\quad\; \hf_\e[u_n-u_m]\to0\quad\;\text{as}\;\; n,m\to\infty
\end{equation}
for some $u\in L_2(\Om_\e)$. In view of (\ref{5.28}), this immediately implies that
\begin{equation}\label{5.30}
\|u_n^\bot-u_m^\bot\|_{\H^1(\Om_\e)} \to0,\quad\;  \la u_n-u_m\ra_G\to 0 \quad\;\text{as}\;\; n,m\to\infty.
\end{equation}
and taking (\ref{5.1b}) and (\ref{5.29}) into account, we then conclude that
\begin{equation}\label{5.31}
\|u_n-u_m\|_{\H^1(\Om_\e)} \to0 \quad\,\text{as}\;\; n,m\to\infty.
\end{equation}
This $u_n$ converges in $\H^1(\Om_\e)$ and due to the first claim in (\ref{5.29}), the limiting function is $u$ which means that $u\in\H^1(\Om_\e)$. By definition (\ref{5.11}) of the form $\hf_\e$ together with (\ref{5.29}), (\ref{5.31}) this implies
\begin{equation*}
\hf_\Om[(1-\chi_\Om)(u_n-u_m)]\to0,\quad\; \|(1-\chi_\Om)u_n-(1-\chi_\Om)u\|_{L_2(\Om_\e)}\to0 \quad\;\text{as}\;\; n,m\to\infty.
\end{equation*}
Since the form $\hf_\Om$ is closed, it follows that $(1-\chi_\Om)u_n$ converges to $(1-\chi_\Om)u$ with respect to the norm in the subspace $\Dom(\hf_\Om)$ of the Sobolev space $\H^1(\Om)$. Consquently, $(1-\chi_\Om)u\in \Dom(\hf_\Om)$, and in view of (\ref{5.21}) we may conclude that $u\in\Dom(\hf_\Om)$, and also $\hf_\e[u_n-u]\to0$ as $n\to\infty$.  This means that the form $\hf_\e$ is closed.

This brings us to the desired conclusion: the operator $\Op_\e$ is associated with a closed symmetric sesquilinear form semibounded from below, and therefore it is self-adjoint.

\subsection{Self-adjointness of the operator $\Op_{0,\b}$}\label{ss:SA0}

By definition, the domain of the adjoint operator $\Op_{0,\b}^*$ consists of all $v\in L_2(\Om)$ such that there exists a function $g\in L_2(\Om)$ obeying the identity
\begin{equation}\label{5.32}
(\Op_{0,\b}u,v)_{L_2(\Om)}=(u,g)_{L_2(\Om)}\quad\text{for all}\;\; u\in\Dom(\Op_{0,\b}),\quad\; \Op_{0,\b}^*v=g.
\end{equation}
Since $u=u_0+(\b-a)^{-1} u_0(x_0)G$, $u_0\in\Dom(\Op_\Om)$, we can rewrite the above identity as
\begin{equation*}
\big(\Op_\Om u_0-c_1(\b-a)^{-1} u_0(x_0)G,v\big)_{L_2(\Om)}=(u_0,g)_{L_2(\Om)} + (\b-a)^{-1} u_0(x_0)(G,g)_{L_2(\Om)}
\end{equation*}
and hence,
\begin{equation}\label{5.34}
(\Op_\Om u_0,v)_{L_2(\Om)}-(\b-a)^{-1} u_0(x_0) (G,c_1 v+g)_{L_2(\Om)} = (u_0,g)_{L_2(\Om)}.
\end{equation}
Proceeding as in the proof of Lemma~\ref{lm:G}, cf. (\ref{3.27})--(\ref{3.28}),
it is straightforward to check that
\begin{align*}
(f,G)_{L_2(\Om^\d)}&+(\b-a)^{-1}(\l+c_1) u_0(x_0) \|G\|_{L_2(\Om^\d)}^2
\\
&+(\l+c_1)(u_0,G)_{L_2(\Om^\d)} = -\!\!\int\limits_{\{x:\, |y|=\d\}} \left(\overline{G}\,\frac{\p u_0}{\p\mathrm{n}}-u_0 \frac{\p\overline{G}}{\p\mathrm{n}}\right)\,\mathrm{d}s.
\end{align*}
Passing to the limit as $\d\to+0$ in the above identity, the left-hand side converges to the analogous expression with the scalar product referring to $L_2(\Om)$. In view of Lemmata~\ref{lm:bnd},~\ref{lm4.0} with $\om_\e$ replaced by $\{x:\, |y|<\d\}$ together with the asymptotics~(\ref{3.0}) and identity (\ref{5.5}), we also get
\begin{equation*}
\lim\limits_{\d\to+0}\int\limits_{\{x:\, |y|=\d\}}  \overline{G}\,\frac{\p u_0}{\p\mathrm{n}}\,\mathrm{d}s=0,\quad\; \lim\limits_{\d\to+0} \int\limits_{\{x:\, |y|=\d\}}  u_0 \frac{\p \overline{G}}{\p\mathrm{n}}\,\mathrm{d}s =-\pi v(x_0)\tr\mathrm{A}.
\end{equation*}
Recalling the definition of $u_0$, the limit $\d\to+0$ thus yields
\begin{equation}\label{5.33}
-\pi u_0(x_0)\tr\mathrm{A}=((\Op_\Om+c_1)u_0,G)_{L_2(\Om)}
\end{equation}
which allows us to rewrite (\ref{5.34}) as
\begin{equation*}
(\Op_\Om u_0,v)_{L_2(\Om)}-(\b-a)^{-1}\kappa ((\Op_\Om+c_1)u_0,G)_{L_2(\Om)}=(u_0,g)_{L_2(\Om)},\quad\; \kappa:=-\frac{(G,c_1v+g)_{L_2(\Om)}}{\pi\tr\mathrm{A}},
\end{equation*}
or equivalently as
\begin{equation*}
(\Op_\Om u_0,v-(\b-a)^{-1}\overline{\kappa} G)_{L_2(\Om)}=(u_0,g+(\b-a)^{-1} c_1\overline{\kappa} G)_{L_2(\Om)}.
\end{equation*}
Since the operator $\Op_\Om$ is self-adjoint, the above identity implies that
\begin{equation}\label{5.35}
w:=v-(\b-a)^{-1} \overline{k} G\in \Dom(\Op_\Om),\qquad \Op_{0,\b}^*w=g+(\b-a)^{-1} c_1\overline{\kappa} G.
\end{equation}
Using then the identity (\ref{5.33}) with $u_0$ replaced by $w$, we get
\begin{equation*}
-\pi w(x_0)\tr\mathrm{A}=  ((\Op_\Om+c_1)w,G)_{L_2(\Om)}=(g+c_1 v,G)_{L_2(\Om)}=-\pi\overline{\kappa}\tr \mathrm{A},
\end{equation*}
and therefore, by virtue of (\ref{5.35}),
\begin{equation*}
v=w+(\b-a)^{-1} w(x_0)G,\quad\; \Op_{0,\b}^* w=g+c_1 (\b-a)^{-1} w(x_0)G,\quad\; w\in\Dom(\Op_\Om),
\end{equation*}
which means that $\Op_{0,\b}^*=\Op_{0,\b}$.

\subsection{Resolvent convergence}\label{ss:ResConv}

Since both the operators $\Op_\e$ and $\Op_{0,\b}$ are self-adjoint, their resolvents are well defined for $\l$ away from the real axis, $\IM\l\ne0$. We choose an arbitrary $f\in L_2(\Om)$ and denote $u_0:=(\Op_{0,\b}-\l)^{-1}f$, $u_\e:=(\Op_\e-\l)^{-1}f$, where in the latter definition the resolvent is applied to the restriction of the function $f$ to $\Om_\e$; with an abuse of notation we keep the same symbol for it. We put $v_\e:=u_\e-u_0$. This function obviously belongs to $W_2^2(\Om_\e)$ and solves the boundary-value problem
\begin{equation*}
(\hat{\Op}-\l)v_\e=0\quad\text{in}\;\;\Om_\e, \quad\; \cB v_\e=0\quad\text{on}\;\; \p\Om,\quad\; \frac{\p v_\e}{\p\mathrm{n}}=\e^{-1}\a v_\e+g_\e\quad\text{on}\;\; \p\om_\e,
\end{equation*}
where
\begin{equation}\label{4.2}
g_\e:=\left(\frac{\p\ }{\p\mathrm{n}}-\e^{-1}\a \right)u_0.
\end{equation}
The corresponding integral equation reads
\begin{equation}\label{4.3}
\hf_\e[v_\e]-\l\|v_\e\|_{L_2(\Om_\e)}^2=-(g_\e,v_\e)_{L_2(\p\om_\e)}.
\end{equation}

The next step is to estimate the right-hand side in (\ref{4.3}). Since $u_0\in \Dom(\Op_{0,\b})$ according (\ref{2.15}), it can be represented as
\begin{equation}\label{4.5}
u_0(x)=v_0(x)+(\b-a)^{-1} v_0(x_0) G(x),\quad\; v_0\in \H^2(\Om),
\end{equation}
and
\begin{equation*}
f=(\Op_{0,\b}-\l)u_0=(\Op_\Om-\l)v_0- (\b-a)^{-1}(\l+c_1) v_0(x_0)G.
\end{equation*}

\begin{lemma}\label{lm4.1}
The inequality
\begin{equation*}
\|v_0\|_{\H^1(\Om)} + \|v_0\|_{\H^2(B_{2R_2}(x_0))} + |v_0(x_0)|\leqslant C\|f\|_{L_2(\Om)}
\end{equation*}
holds, where $C$ is a constant  independent of $f$ but in general depending on $\l$.
\end{lemma}
\begin{proof}
Throughout the proof the symbol $C$ stands again for various inessential constants independent of $v_0$. Since the operator $\Op_{0,\b}$ is self-adjoint and $\l\notin\spec(\Op_{0,\b})$, we immediately get
\begin{equation}\label{4.8}
\begin{aligned}
\|v_0\|_{L_2(\Om)}^2 &+ 2\RE\overline{(\b-a)^{-1} v_0(x_0)} (v_0,G)_{L_2(\Om)}
 \\
 &+ |v_0(x_0)|^2 \|G\|_{L_2(\Om)}^2=\|u_0\|_{L_2(\Om)}^2\leqslant \frac{\|f\|_{L_2(\Om)}^2 }{\dist(\l,\spec(\Op_{0,\b}))^2} .
\end{aligned}
\end{equation}
We observe that the function $v_0$ solves the operator equation
\begin{equation}
(\Op_\Om-\l)v_0=f+(\b-a)^{-1}(\l+c_1) v(x_0)G\quad\text{in}\;\;\Om\label{4.18}.
\end{equation}
Repeating the steps that led  us to identity (\ref{5.33}), we can confirm that
\begin{equation*}
(f,G)_{L_2(\Om)}+(\l+c_1)(u_0,G)_{L_2(\Om)} = -\pi v(x_0)\tr\mathrm{A}.
\end{equation*}
In view of (\ref{4.8}), this implies that
\begin{equation}\label{4.16}
|v_0(x_0)|\leqslant C\|f\|_{L_2(\Om)},
\end{equation}
and by Cauchy-Schwarz inequality we then find
\begin{equation*}
\Big|\overline{(\b-a)^{-1} v_0(x_0)} (v_0,G)_{L_2(\Om)}  \Big|\leqslant |(\b-a)^{-1}||v_0(x_0)|\|v_0\|_{L_2(\Om)}\|G\|_{L_2(\Om)} \leqslant \frac{1}{2}\|v_0\|_{L_2(\Om)}^2+C\|f\|_{L_2(\Om)}^2.
\end{equation*}
This estimate in combination with (\ref{4.8}) yields
\begin{equation}\label{4.19}
\|v_0\|_{L_2(\Om)}^2\leqslant C\|f\|_{L_2(\Om)}^2,
\end{equation}
thus by (\ref{4.16}) we infer that the right-hand side in (\ref{4.18}) can be estimated as
\begin{equation*}
\|f+(\b-a)^{-1}(\l+c_1) v(x_0)G\|_{L_2(\Om)}\leqslant C\|f\|_{L_2(\Om)}.
\end{equation*}
It follows then from (\ref{4.18}) that
\begin{equation*}
\|v_0\|_{\H^1(\Om)} + \|v_0\|_{\H^2(B_{2R_2}(x_0))} +\|v_0\|_{\H^2(\Om)}\leqslant C\|f\|_{L_2(\Om)}
\end{equation*}
and his estimate together with (\ref{4.19}) completes the proof.
\end{proof}

\medskip

Recalling (\ref{5.1a}), (\ref{5.1b}), we represent functions $v_\e$ and $v_0$ as
\begin{equation}\label{4.22}
 v_{0,\bot}(x):=v_0(x)-\la v_0\ra_{\p\om_\e}, \quad\; v_\e=v_\e^\bot+\la v_\e\ra_G G,\quad\;
 \int\limits_{\p\om_\e} v_{0,\bot} \,\mathrm{d}s=\int\limits_{\p\om} v_\e^\bot\,\mathrm{d}s=0.
\end{equation}
Furthermore, in view of (\ref{4.2}) and (\ref{4.5}) the function $g_\e$ has the following representation,
\begin{equation}\label{4.17}
\begin{gathered}
g_\e=g_{\e,1}+g_{\e,2}+g_{\e,3},\quad\;
g_{\e,1}:=\frac{\p v_0}{\p \mathrm{n}},
\\g_{\e,2}:=\e^{-1} (v_0 -v_0(x_0))\a,\quad\;  g_{\e,3}:=v_0(x_0)(\b-a)^{-1}\left( \frac{\p G}{\p\nu}-\e^{-1}\a G\right)-\e^{-1}\a v_0(x_0).
\end{gathered}
\end{equation}
We have
\begin{equation*}
(g_{\e,1},v_\e)_{L_2(\p\om_\e)}=(g_{\e,1},v_\e^\bot)_{L_2(\p\om_\e)}
+\overline{
\la v_\e\ra_G}\,(g_{\e,1},G)_{L_2(\p\om_\e)},
\end{equation*}
and therefore from Lemmata~\ref{lm4.1},~\ref{lm:bnd}, and~\ref{lm:bnd-mean} we infer that
\begin{equation}\label{4.21}
\begin{aligned}
\big|(g_{\e,1},v_\e)_{L_2(\p\om_\e)}\big|\leqslant &\: C\e|\ln\e| \|v\|_{\H^2(B_{2R_2}(x_0))}\big(\|v_\e^\bot\|_{\H^1(\Om_\e)}+|\la v_\e\ra_{G}|\big)
\\[.3em]
\leqslant &\: C\e|\ln\e|\|f\|_{L_2(\Om)}\big(\|v_\e^\bot\|_{\H^1(\Om_\e)}+|\la v_\e\ra_{G}|\big).
\end{aligned}
\end{equation}
As before, the symbol $C$ stands for inessential constants independent of $\e$, $f$, $v_0$, $v_\e$, and $x$.

In view of the decomposition (\ref{4.22}), the function $g_{\e,2}$ can be represented as
\begin{equation*}
g_{\e,2}=\e^{-1} (v_{0,\bot}+ g_{\e,4})\a,\quad\; g_{\e,4}:=\la v_0\ra_{\p\om_\e}-v_0(x_0),
\end{equation*}
and using Lemmata~\ref{lm:bnd},~\ref{lm:bnd-mean},~\ref{lm4.0}, and~\ref{lm4.1}, we obtain
\begin{equation}\label{4.25}
\begin{aligned}
\big|(g_{\e,2},v_\e)_{L_2(\p\om_\e)}\big|\leqslant & C\e^{-1}|\ln\e|^{-1} \big(\|v_{0,\bot}\|_{L_2(\p\om_\e)}+\e^\frac{1}{2}|g_{\e,4}|\big) \|v_\e\|_{L_2(\p\om_\e)}
\\[.3em]
\leqslant & C\e^\frac{1}{2} \|f\|_{L_2(\Om)}\big(\|v_\e^\bot\|_{\H^1(\Om_\e)}+|\la v_\e\ra_{G}|\big).
\end{aligned}
\end{equation}
Let us proceed to assessment of the scalar product $(g_{\e,3},v_\e)_{L_2(\p\om_\e)}$. Using representation (\ref{4.17}) together with (\ref{5.8}), (\ref{2.4}), we get
\begin{equation}\label{4.38}
(g_{\e,3},v_\e)_{L_2(\p\om_\e)} = (g_{\e,3},v_\e^\bot)_{L_2(\p\om_\e)}+\la v_\e\ra_{G} (g_{\e,3},G)_{L_2(\p\om_\e)}
\end{equation}
and
\begin{equation*}
(g_{\e,3},G)_{L_2(\p\om_\e)}=v_0(x_0) \big((\b-a)^{-1}(K+\pi a\tr\mathrm{A})+\pi\tr \mathrm{A}+\mathcal{O}(\ln^{-1}\e)\big).
\end{equation*}
In view of (\ref{2.5}) and Lemma~\ref{lm4.1} we thus have
\begin{equation}\label{4.31}
\big|\la v_\e\ra_{G}(g_{\e,3},G)_{L_2(\p\om_\e)}\big|\leqslant C|\ln\e|^{-1}\|f\|_{L_2(\p\Om)}\big|\la v_\e\ra_{G}\big|.
\end{equation}
Next we use identities (\ref{5.2}), (\ref{2.2c}) and Lemmata~\ref{lm:bnd-mean},~\ref{lm4.1} to estimate the first term on the right-hand side of (\ref{4.38}),
\begin{align*}
\big|(g_{\e,3},v_\e^\bot)_{L_2(\p\om_\e)}\big| = &\:\e^{-1} |\ln\e|^{-1} \Big|\big(\Phi_1-\a_0+\ln^{-1}\e\,(\Phi_2-\a_1),v_\e^\bot \big)_{L_2(\p\om_\e)}\Big|
\\
\leqslant &\: C|\ln\e|^{-1}\|f\|_{L_2(\Om)} \big(\|v_\e^\bot\|_{\H^1(\Om_\e)}+|\la v_\e\ra_{G}|\big).
\end{align*}
This estimate and (\ref{4.31}) lead us to a bound for $(g_{\e,3},v_\e)_{L_2(\p\om_\e)}$,
\begin{equation*}
\big|(g_{\e,3},v_\e)_{L_2(\p\om_\e)}\big| \leqslant C|\ln\e|^{-1} \|f\|_{L_2(\Om)} \big(\|v_\e^\bot\|_{\H^1(\Om_\e)}+|\la v_\e\ra_{G}|\big).
\end{equation*}
then (\ref{4.25}), (\ref{4.21}), and (\ref{4.17}) imply the final estimate for the right-hand side in (\ref{4.3}),
\begin{equation}\label{4.33}
\big|(g_\e,v_\e)_{L_2(\p\om_\e)}\big|\leqslant C
|\ln\e|^{-1}
\|f\|_{L_2(\Om_\e)}\big(\|v_\e^\bot\|_{\H^1(\Om_\e)}+|\la v_\e\ra_{G}|\big).
\end{equation}

Now we consider separately the imaginary and real part of the both sides of equation (\ref{4.3}), then using (\ref{5.28}) we arrive at
\begin{equation}\label{4.33a}
\begin{aligned}
&\|v_\e\|_{L_2(\Om_\e)}^2\leqslant C
|\ln\e|^{-1}
\|f\|_{L_2(\Om_\e)}\big(\|v_\e^\bot\|_{\H^1(\Om_\e)}+|\la v_\e\ra_{G}|\big),
\\[.3em]
& \|v_\e^\bot\|_{\H^1(\Om_\e)}^2+|\la v_\e\ra_{G}|^2
 \leqslant C
|\ln\e|^{-1}
\|f\|_{L_2(\Om_\e)}\big(\|v_\e^\bot\|_{\H^1(\Om_\e)}+|\la v_\e\ra_{G}|\big),
\end{aligned}
\end{equation}
where the second estimate implies
\begin{equation}\label{4.40}
\|v_\e^\bot\|_{\H^1(\Om_\e)} +|\la v_\e\ra_{G}|
\leqslant C
|\ln\e|^{-1}
\|f\|_{L_2(\Om_\e)}.
\end{equation}
In this way we get the inequality
\begin{equation}\label{4.41}
\|v_\e\|_{L_2(\Om_\e)} \leqslant  \|v_\e^\bot\|_{L_2(\Om_\e)} + |\la v_\e\ra_{G}| \|G\|_{L_2(\Om_\e)}\leqslant C
|\ln\e|^{-1}
\|f\|_{L_2(\Om_\e)}
\end{equation}
which proves the convergence (\ref{cnv1}).

As for the second claim of Theorem~\ref{th:main}, using asymptotics (\ref{3.0}) it is easy to check that
\begin{equation}\label{4.45}
\|\nabla G\|_{L_2(\Om_0\setminus\om_\e)}\leqslant  C|\ln\e|^\frac{1}{2},\quad\; \|\nabla (1-\chi_\Om) G\|_{L_2(\Om_\e)} + \hf_\Om[(1-\chi_\Om)G]\leqslant C.
\end{equation}
and consequently, by virtue of (\ref{4.40}),
\begin{equation}\label{4.46}
\|\nabla v_\e\|_{L_2(\Om_\e)}^2 \leqslant 2\|\nabla v_\e^\bot\|_{L_2(\Om_\e)}^2 + 2|\la v_\e\ra_G|^2 \|\nabla G\|_{L_2(\Om_\e)}^2 \leqslant C|\ln\e|^{-1}\|f\|_{L_2(\Om)}^2.
\end{equation}
This inequality in combination with (\ref{4.41}) proves (\ref{cnv2}).

Let us pass to the last claim. It follows from the estimate (\ref{5.28a}) and identity (\ref{4.3}) that
\begin{equation*}
\hf_\Om[(1-\chi_\Om)v_\e^\bot]+c_5\|v_\e^\bot\|_{L_2(\Om_\e)}^2\leqslant C \|v_\e\|_{L_2} +\big|(g_\e,v_\e)_{L_2(\Om_\e)}\big|  + \hf_{\Om_0}[(1-\chi_\Om)v_\e^\bot].
\end{equation*}
Using now (\ref{4.33}), (\ref{4.40}), and (\ref{4.41}), we obtain
\begin{equation*}
\hf_\Om[(1-\chi_\Om)v_\e^\bot]\leqslant C|\ln\e|^{-2}\|f\|_{L_2(\Om)}^2
\end{equation*}
and by (\ref{4.45}) and (\ref{4.41}) this implies that
\begin{equation*}
\hf_\Om[(1-\chi_\Om)v_\e]+\|(1-\chi_\Om)v_\e\|_{L_2(\Om_\e)}^2\leqslant C|\ln\e|^{-2}\|f\|_{L_2(\Om)}^2.
\end{equation*}
Together with (\ref{4.46}) and (\ref{4.41}), the above inequality leads us to (\ref{cnv3}).

Let us finally demonstrate that the estimates (\ref{cnv1}),  (\ref{cnv2}), and (\ref{cnv3}) are order sharp. To this aim, it is sufficient to consider a suitable particular case, for instance,
\begin{equation*}
\Om=\mathbb{R}^2,\qquad x_0=0,\qquad \hat{\Op}=-\D,\qquad c_1=1,\qquad \mathrm{A}=\mathrm{E},\qquad
\Om_0:=B_1(0).
\end{equation*}
 The function $G$ can be then found explicitly,
\begin{equation*}
G(x)=\frac{\pi}{2\iu} H_0(\iu|x|),
\end{equation*}
where $H_0$ is the Hankel function of the first kind. For the `hole' we choose the disc of radius $b$, that is, $\om:=B_b(0)$. Then according to (\ref{2.4}), the function $\a_0$ is constant, $\a_0=-b^{-1}$, on the hole parimeter, and we choose $\a_1$ being a constant as well. The asymptotics of $G$ is well known,
\begin{equation*}
G(x)=\ln|x|+a+\mathcal{O}\big(|x|^2\ln|x|\big),\quad |x|\to0,\quad\, a:=\g-\ln 2, \quad\; \g:=\lim\limits_{n\to+\infty} \bigg(\sum\limits_{m=1}^{n}\frac{1}{m}-\ln n\bigg).
\end{equation*}
The constants $ K$ and $\b$ defined in (\ref{2.4a}), (\ref{2.5}) are in this case the following,
\begin{equation*}
 K=2\pi\big(\ln b-b\a_1\big),\quad\; \b=b\a_1-\ln b.
\end{equation*}
We also observe that in terms of the standard definition of the point interaction, the above operator coincides with $-\D_{\z,x_0}$ introduced in \cite[Thm.~I.5.3]{AGHH}, referring to the coupling constant $2\pi\z=-b$.  The hole radius $b$ is positive by definition, so in this case we are able determine explicitly the range of the coupling strengths for which our approximation works.

Let $v_0\in C_0^\infty(\mathbb{R}^2)$ be a non-vanishing radially symmetric function such that $v_0(0)\ne0$. Then the function $u_0(x):=v_0(x)+(\b-a)^{-1} v_0(0) G(x)$ is in the domain of the operator $\Op_{0,\b}$ and
\begin{equation*}
(\Op_{0,\b}-\l)u_0=f:=-\D v_0 +(\b-a)^{-1}(1-\l)v_0(x_0)G
\end{equation*}
for each $\l=k^2$ with $\IM k>0$, $\IM\l\ne0$. It follows that the function $v_\e:=(\Op_\e-\l)^{-1}f-u_0$ solves the boundary-value problem
\begin{align*}
&(-\D-\l)v_\e=0\quad\text{in}\;\; \mathbb{R}^2\setminus B_{b\e}(0), \quad\; -\frac{\p v_\e}{\p |x|}+ \frac{1}{\e\ln\e} \left(\frac{1}{b}-\frac{\a_1}{\ln\e}\right)v_\e=h_\e\quad\text{on}\;\; \p B_{b\e}(0),
\\
&h_\e:=\left(\frac{\p u_0}{\p |x|}-\frac{1}{\e\ln\e} \left(\frac{1}{b}-\frac{\a_1}{\ln\e}\right)u_0\right)\Bigg|_{|x|=b\e}
\end{align*}
and can be found explicitly:
\begin{equation*}
v_\e(x)=\frac{h_\e}{c_\e} H_0(\iu k |x|), \qquad c_\e:=\left(-\frac{\p\ }{\p r}+\frac{1}{\e\ln\e} \left(\frac{1}{b}-\frac{\a_1}{\ln\e}\right)\right)H_0(\iu k|x|)\Bigg|_{r=b\e}.
\end{equation*}
With the explicit formul{\ae} for all the considered functions in hand, we can find the asymptotics of the quotient $h_\e/c_\e$,
\begin{equation*}
\frac{h_\e}{c_\e}=\frac{\iu\pi}{2} \frac{v_0(x_0)\a_1^2 b^2(\b-a)^{-2}}{(1-(\b-a)^{-1}\ln k)\ln\e}+ \mathcal{O}(\ln^{-2}\e),
\end{equation*}
which means that
\begin{equation*}
\|v_\e\|_{L_2(\Om_\e)}\geqslant C|\ln\e|^{-1},\quad\;
\|\nabla v_\e\|_{L_2(\Om_\e)}\geqslant C|\ln\e|^{-\frac{1}{2}},
\end{equation*}
where $C$ is a positive constant independent of $\e$. Consequently, the estimates (\ref{cnv1}),  (\ref{cnv2}), and (\ref{cnv3}) are sharp up to a multiplicative constant. This concludes the proof of Theorem~\ref{th:main}.

\subsection{Spectral convergence}

In this subsection we prove  Theorem~\ref{th2.2}.  We employ the ideas proposed in the proof of a similar statement in \cite{PRSE}, see Theorem~2.5 and Section~7  in the cited work.

The proof is based on standard results on the convergence of  spectra and associated spectral projectors  with respect to the  resolvent norm, see, for instance, \cite[Thm. VIII.23]{RS}. However, we can not apply directly this theorem since our operators $\Op_\e$ and $\Op_{0,\b}$ act in different Hilbert spaces, $L_2(\Om_\e)$ and $L_2(\Om_0)$. To overcome this obstacle, we introduce an auxiliary multiplication operator in $L_2(\om_\e)$ acting as $\Op_{\om_\e} u:=\e^{-1} u$. This simple operator is self-adjoint, its spectrum consists of the only eigenvalue $\l=\e$ of an infinite multiplicity and the resolvent satisfies the relation
\begin{equation}\label{6.1}
\|(\Op_{\om_\e}-\l)^{-1}\|_{L_2(\Om_\e)\to L_2(\Om_\e)}=\frac{\e}{|1-\e\l|},\quad\; \l\ne\e.
\end{equation}
In view of  Lemma~\ref{lm4.1} and estimate~(\ref{3.3a}) in Lemma~\ref{lm:bnd} we have an obvious estimate,
\begin{equation}\label{6.2}
\|(\Op_{0,\b}-\l)^{-1}\|_{L_2(\Om)\to L_2(\om_\e)}\leqslant C\e|\ln\e|^{\frac{1}{2}},
\end{equation}
 valid for all $\l$ with a non-zero imaginary part,  where $C$ is a constant independent of $\e$ but depending on $\l$.

We regard the space $L_2(\Om)$ as the direct sum $L_2(\Om)=L_2(\Om_\e)\oplus L_2(\om_\e)$ and  consider the direct sum $\tilde{\Op}_\e:=\Op_\e\oplus \Op_{\om_\e}$. Then estimates~(\ref{cnv1}) and (\ref{6.1}),~(\ref{6.2}) imply that
\begin{equation}\label{6.3}
\begin{aligned}
\|(&\tilde{\Op}_\e-\l)^{-1}-(\Op_{0,\b}-\l)^{-1}\|_{L_2(\Om)\to L_2(\Om)}\leqslant
\|(\Op_\e-\l)^{-1}-(\Op_{0,\b}-\l)^{-1}\|_{L_2(\Om)\to L_2(\Om_\e)}
\\[.3em]
&+ \|(\Op_{\om_\e}-\l)^{-1}\|_{L_2(\om_\e)\to L_2(\om_\e)} + \|(\Op_{0,\b}-\l)^{-1}\|_{L_2(\Om)\to L_2(\om_\e)}
\leqslant  C|\ln\e|^{-1}
\end{aligned}
\end{equation}
for $\IM\l\ne0$, where $C$ is a constant independent of $\e$ but depending on $\IM\l$. Now we  apply  Theorem~VIII.23 from \cite{RS}  to conclude that the spectrum of the operator $\tilde{\Op}_\e$ converges to that of the operator $\Op_{0,\b}$. Since the spectrum  of $\Op_{\om_\e}$ consists of the only point $\l=\e^{-1}$, which escapes to the infinity as $\e\to+0$, and
\begin{equation}\label{6.4}
\spec(\tilde{\Op}_\e)=\spec(\Op_\e)\cup\{\e^{-1}\},
\end{equation}
we obtain the stated convergence of the spectrum of the operator $\Op_\e$. The  convergence of the spectral projections  corresponding to any interval $[\vr_1,\vr_2]$ with $\vr_1$ and $\vr_2$ from the resolvent set of $\Op_{0,\b}$ also  follows from Theorem~VIII.23 in \cite{RS}.

Let us next prove inclusion (\ref{2.17}). We choose an arbitrary but fixed segment $Q:=[\vr_1,\vr_2]$ and consider  $\l\in\mathbb{C}$ such $\l=t+\iu  |\ln\e|^{-1}$ with  $t\in Q\cap\spec(\Op_{0,\b})$; the set of such $\l$ is denoted by $Q_\e$. For $\l\in Q_\e$ we recall the well-known formul{\ae}
\begin{align*}
&\big\|(\Op_{0,\b}-\l)^{-1}\big\|_{L_2(\Om)\to L_2(\Om)}=\frac{1}{\dist(\l,\spec(\Op_{0,\b}))},
\\
&\big\|(\tilde{\Op}_\e-\l)^{-1}\big\|_{L_2(\Om)\to L_2(\Om)}=\frac{1}{\dist(\l,\spec(\tilde{\Op}_\e))}=\frac{1}{\dist(\l,\spec(\Op_\e))},
\end{align*}
where in the latter identity we have also employed   (\ref{6.4}). These  relations  and estimate (\ref{6.3}) imply  that
\begin{equation*}
\left|\frac{1}{\dist(\l,\spec(\Op_\e))} -\frac{1}{\dist(\l,\spec(\Op_{0,\b}))}\right|\leqslant C|\ln\e|^{-1},
\end{equation*}
and hence, for $\l\in Q_\e$,
\begin{equation*}
\frac{1}{\dist(\l,\spec(\Op_\e))}\geqslant \frac{1}{\dist(\l,\spec(\Op_{0,\b}))}-C|\ln\e|^{-1}
\geqslant  |\ln\e| -C|\ln\e|^{-1}\geqslant \frac{|\ln\e|}{2},
\end{equation*}
 in other words, 
\begin{equation*}
\dist(\l,\spec(\Op_\e)) \leqslant 2 |\ln\e|^{-1}\quad\text{as}\quad \l\in Q_\e.
\end{equation*}
Hence the distance from the set $\spec(\Op_\e)\cap Q$ to the set $\spec(\Op_{0,b})\cap Q$ does not exceed $2|\ln\e|^{-1}$ and this proves inclusion (\ref{2.17}).

Finally, let $\l_0$ be  an isolated eigenvalue of the operator $\Op_{0,\b}$  of multiplicity $n$ and $\cP_{0,\b}$ be the projection on the associated eigenspace in $L_2(\Om)$. Then the above proven facts imply immediately that there exist exactly $n$ isolated eigenvalues of the operator $\Op_\e$ converging to $\l_0$, naturally with the multiplicities taken into account; we refer to them as to perturbed eigenvalues. By $\cP_\e$ we denote the total projection associated with them. Inclusion (\ref{2.17}) ensures that the distance from the  perturbed eigenvalues to $\l_0$ is estimated by $C|\ln\e|^{-1}$ with some constant $C$ independent of $\e$. We fix $\d>0$ such that the ball $B_\d(\l_0)$ in the complex plane contains no other points of spectra of $\Op_\e$ and $\Op_{0,\b}$ except for $\l_0$ and the  perturbed eigenvalues. Then we know that
\begin{equation*}
\cP_\e=\frac{1}{2\pi\iu}\int\limits_{\p B_\d(\l_0)} (\tilde{\Op}_\e-\l)^{-1}\di\l=\frac{1}{2\pi\iu}\int\limits_{\p B_\d(\l_0)} (\Op_\e-\l)^{-1}\di\l,\quad\;
\cP_{0,\b}=\frac{1}{2\pi\iu}\int\limits_{\p B_\d(\l_0)} (\Op_{0,\b}-\l)^{-1}\di\l,
\end{equation*}
and  consequently,
\begin{equation}\label{6.5}
\cP_\e-\cP_0=\frac{1}{2\pi\iu}\int\limits_{\p B_\d(\l_0)}
\big((\Op_\e-\l)^{-1}-(\Op_{0,\b}-\l)^{-1}\big)\di\l.
\end{equation}
Since the contour $\p B_\d(\l_0)$ is separated from the spectra of both operators $\Op_{0,\b}$ and $\Op_\e$, estimates (\ref{cnv1}), (\ref{cnv2}), (\ref{cnv3}) remain true also for $\l\in B_\d(\l_0)$. Indeed, one can reproduce literally the  argumentation in Section~\ref{ss:ResConv} because  the fact that $\IM\l$ is non-zero was employed only in Lemma~\ref{lm4.1} and in (\ref{4.33a}); both this lemma and  the  inequalities obviously remain true in our case. Now the  desired estimates for the spectral projections follow  from  identity (\ref{6.5}) and  estimates (\ref{cnv1}), (\ref{cnv2}), (\ref{cnv3}). This completes the proof of Theorem~\ref{th2.2}.

\subsection*{Acknowledgements}

The work of P.E. was supported by the European Union within the project CZ.02.1.01/0.0/0.0/16 019/0000778.


\end{document}